\documentclass[sigconf,natbib=false]{acmart}

\RequirePackage[abbreviate=true, dateabbrev=true, natbib=true, isbn=false, doi=false, urldate=comp, url=false, maxbibnames=9, maxcitenames=1, backref=false, backend=biber, style=ACM-Reference-Format, language=american,firstinits=true]{biblatex}

\usepackage{filecontents}
\usepackage{booktabs} 
\usepackage{arydshln} 
\usepackage{graphicx}
\graphicspath{ {images/} }
\usepackage{pdfpages}
\usepackage{blindtext}
\usepackage{subcaption}
\usepackage{amsmath}
\usepackage{bm} 
\usepackage{amsmath, mathrsfs,amsfonts,amssymb}
\usepackage{wasysym} 
\usepackage{xfrac} 

\usepackage{savesym} 
\savesymbol{openbox}
\usepackage{amsthm} 
\usepackage{thmtools}

\usepackage{commath}
\usepackage[export]{adjustbox}

\usepackage{float}

\usepackage{stfloats}
\usepackage{url}
\usepackage{cleveref} 
\usepackage{algorithm}
\usepackage[noend]{algpseudocode}

\begin{document}

\title{Mining Hidden Populations through Attributed Search}

\author{Suhansanu Kumar, Heting Gao, Changyu Wang, Hari Sundaram, Kevin Chen-Chuan Chang} 
    \affiliation{ 
    \department{Department of Computer Science}.
    \institution{University of Illinois, Urbana Chamapaign}
    }
    \email{ {skumar56, hgao17, changyuw, hs1, kcchang}@illinois.edu }

\settopmatter{printacmref=false} 
\renewcommand\footnotetextcopyrightpermission[1]{} 

\begin{abstract}

Researchers often query online social platforms through their application programming interfaces (API) to find target populations such as people with mental illness~\cite{De-Choudhury2017} and jazz musicians~\cite{heckathorn2001finding}. Entities of such target population satisfy a property that is typically identified using an oracle (human or a pre-trained classifier). When the property of the target entities is not directly queryable via the API, we refer to the property as `hidden' and the population as hidden population. 
Finding individuals who belong to these populations on social networks is hard because they are non-queryable, and the sampler has to explore from a combinatorial query space within a finite budget limit. 
By exploiting the correlation between queryable attributes and the population of interest and by hierarchically ordering the query space, we propose a Decision tree-based Thompson sampler (\texttt{DT-TMP}) that efficiently discovers the right combination of attributes to query. 
Our proposed sampler outperforms the state-of-the-art samplers in online experiments, for example by 54\% in Twitter. When the number of matching entities to a query is known in offline experiments, \texttt{DT-TMP} performs exceedingly well by a factor of 0.9-1.5$\times$ over the baseline samplers. 

\end{abstract}


\maketitle
\section{Introduction}
\label{sec:introduction}

Social networks have made available large amounts of public information online. This colossal information has led to the growth of industries such as social media marketing and management, online advertisements. Futher, it has led to scientific advancements in analysis and mining of social information by several organizations including ICWSM and ASONAM. 

However, with the enormous size of crowds on social networks such as Twitter and Facebook, we are facing difficulty in mining information concerning a target group of the population. For example, social scientists and advertisers are increasingly interested in understanding the online behavior of a specific population. Examples include: people with mental illnesses~\cite{De-Choudhury2017}, sex workers~\cite{malekinejad2008using}, cyber bullying~\cite{raisi2017cyberbullying}, hacked accounts~\cite{karimi2018toward} to name a few. More generally, the researchers' goal is to find people that satisfy a certain property, but \textit{crucially}, API doesn't provide direct querying of the hidden property. For example, one cannot use the phrase ``mental illness'' on Twitter to identify people on Twitter potentially suffering from depression because they rarely use that phrase in any of their tweets to self-describe. Thus when we refer to ``hidden populations,'' we are more concretely referring to populations with a non-queryable property.

In this paper, \textit{we focus on identifying hidden populations through attributed search}. Many social networks allow for searching via attributed query, in addition to textual query. For example, we can also specify time and location attributes on Twitter in addition to text. In contrast to attributed search, text search has received considerable attention in the IR community. While there exist significant work on web crawling~\cite{olston2010web, chakrabarti1999focused,alvarez2007crawling}, we are missing such technologies to mine the entities e.g., users on Twitter or GitHub) using their social attributes  (e.g., for Twitter: their tweets, location; for GitHub: programming language). To address this gap, we explore the problem of hidden population sampling from online social platforms using attributed search for the first time. 

There are two prominent reasons why attributed search of hidden populations on social networks is challenging:

\textbf{Combinatorial search space:} A combination of queryable attributes expresses a query issued to the application programming interface (API). Thus, the size of query space is the product of the attribute cardinalities (the number of possible values for each attribute) and grows exponentially as the number of attributes increases. Since social networks employ rate limits affecting the number of queries (e.g., Twitter API allows only 60 API calls in an hour) and search result limits (e.g., Twitter API returns at most 100 entities in single API call) affecting the number of results obtained in a single query, a naive hidden population sampler is likely to remain in exploration phase after it exhausts the query budget.

\textbf{Black-box API:} The internal mechanism of online APIs is often propriety information unavailable to users. The sampler interacts with API using query to get a subset of entities (returned result) matching the query. Some of the challenges associated with an unknown query system are: a) \textit{stochastic feedback} i.e., different pages of a query very likely have different number of hidden entities, b) \textit{re-sampling of a hidden entity}, i.e. the API returns the same hidden entity for two different queries when the entity satisfies both queries, c) \textit{variable size of returned results}, i.e. the API fewer than the page-size of entities when there is not enough matching entities in the population. 
Thus, it becomes difficult for online samplers to efficiently discover high-quality queries that would lead to the sampling of a high number of hidden entities. 

The key insight for hidden population sampling is to identify high-quality queries and issue them multiple times. First, we address the problem of combinatorial search space by hierarchically organizing the query space in the form of a tree. Subsequently, we use a \textit{decision-tree} based search strategy to systematically explore the query space by expanding along high yielding decision-tree branches. Second, we address the problem of black-box API by using the returned set of results to estimate the quality of not just the issued query but also related queries sharing one or more attribute combinations. We employ a reward function to estimate the unique un-sampled hidden entities that can be obtained by issuing a query. Our unified reward function takes into account the stochastic feedback and re-sampling effect while allowing for a exploration-exploitation among queries. We use reinforcement learning based \textit{Thompson sampling} to define the reward function.

We put together the above insights to propose a Decision-Tree Thompson sampling (\texttt{DT-TMP}) algorithm. We illustrate the working of \texttt{DT-TMP} using a 3D toy model in \Cref{fig:toy_model} comprised of three queryable attributes represented by three dimensions. The \texttt{DT-TMP} algorithm maintains a query pool from where it queries the API. We initialize the query pool with the most general query, i.e. the algorithm searches from the entire population. Then, after an epoch time, the algorithm identifies the most promising query (i.e. attribute and the corresponding value)---the one with the highest reward. The query pool expands to add new queries as shown in the second subplot. That is, \texttt{DT-TMP} issues a query with this attribute value as a predicate to identify the next attribute and corresponding attribute value to use as a conjunct. Thus, the algorithm over iterations converges to explore among the high hidden entity density regions of the population. Based on the returned set feedback, \texttt{DT-TMP} updates reward function corresponding to every query in the query pool. Finally, the algorithm terminates when we've exhausted the query budget.

In summary, this paper makes the following contributions:
\begin{itemize}
    \item We propose a novel hidden population sampling problem in online social platforms via attributed search. Further, we propose a reinforcement learning based sampling strategy that performs hierarchical exploration of the combinatorial query space. 
    \item Our generalized sampling strategy applies to diverse web-forms having a variable number and type of attributes with varying attribute cardinalities. 
    \item We perform a comprehensive set of experiments over a suite of twelve sampling tasks on three online web-query platforms: Twitter, RateMDs and GitHub, and three offline entity datasets: Patent, Adult and Auto that illustrates the efficacy of \texttt{DT-TMP} sampler. \texttt{DT-TMP} outperforms all baseline samplers, for example by a margin of 54\% on Twitter. 
    \item We perform an extensive ablation study to understand the impact of different sampling factors. We find page-size, attribute cardinality, the number of queryable attributes and correlation between queryable attributes and the hidden property are the prominent factors that affect sampling. 
\end{itemize}

\begin{figure}
\centering
\includegraphics[width=\linewidth,trim={2cm 11cm 4cm 6cm},clip]{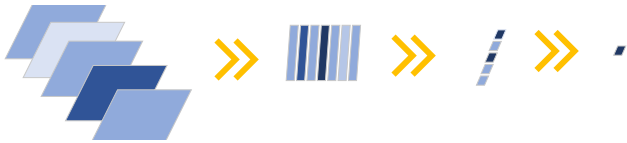}
\caption{\small The entity population is represented by a cube (color indicates the distribution of hidden entities). \texttt{DT-TMP} hierarchically explores the entity population using a drill down approach. It first finds the best query plane (general query), followed by the best line in a query plane (more specific) query and so on.} 
\label{fig:toy_model}
\end{figure}
\section{Problem statement}
\label{sec:problem statement}

In this section, we motivate the problem of hidden population sampling through a real-world example. A formal definition of the sampling problem follows the motivating example. Finally, we present the challenges associated with the problem, and the desired characteristics of an ideal sampler.


Consider a scenario in which a healthcare expert or a researcher is interested in reaching out to the depressed people on Twitter. Assume that the expert has designed a classifier for identifying whether a Twitter user is depressed or not based on the user's profile description and activity~\cite{de2013predicting, tsugawa2015recognizing}. The expert's objective is then to retrieve a maximum number of Twitter users that have the \textit{hidden property} of depression. However, it is not trivial under present circumstances to sample the depressed population from Twitter due to several bottlenecks. The database of Twitter accounts is accessible only through Twitter's application programming interface (API). As shown in~\Cref{fig:twitter_interface}, Twitter allows its user accounts to be queried by only some specific \textit{attributes} such as time, location and text. Besides, Twitter permits only 180 API calls in a 15-minute window. Notice that the depressed population is distributed across the entire Twitter population necessitating the sampler to consider all types of queries. Furthermore, the distribution of the depressed population across different queries is unknown making it difficult to frame queries that would yield in a high discovery of the depressed population within a limited \textit{budget} of API calls.

\begin{figure}
 \centering
 \includegraphics[width=3.2in, height=3.2in]{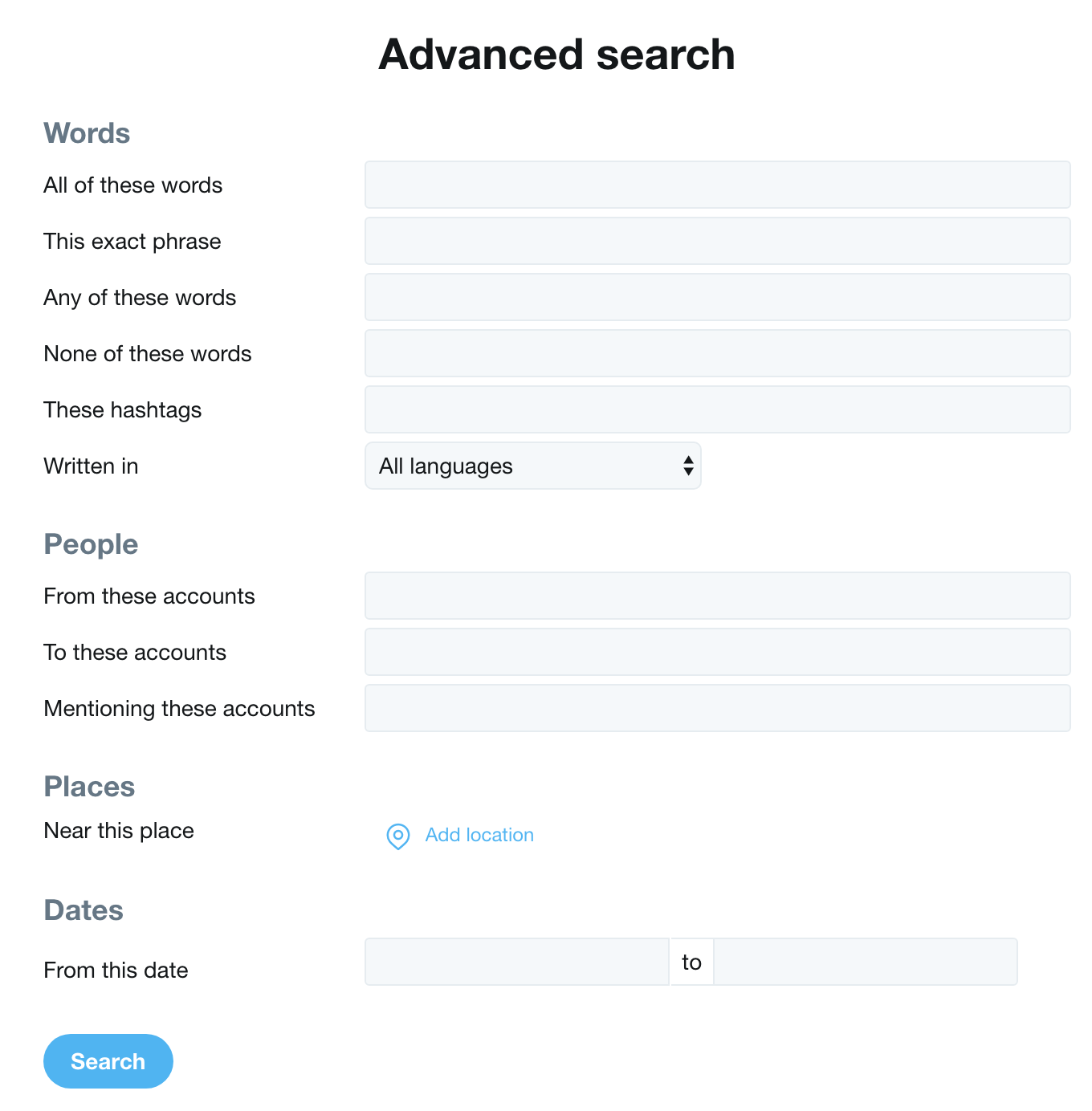}
 \caption{Twitter query interface shown as a typical example of online social platform interface comprising several queryable attributes.}
 \label{fig:twitter_interface}
\end{figure}

To explain the sampling framework, we describe two major features of online social platform services like Twitter API: \textit{query interface} and \textit{returned-result set}. Query interface as shown in~\Cref{fig:twitter_interface} lets the expert query Twitter by setting attribute-values to the queryable attributes. For example, the expert may set the location attribute to `New York' and text-attribute to `mental health' and time attribute to `*' (or ignoring it). In other words, the \textit{query} can be interpreted as a conjunction over queryable attributes say  $A_i$ ($i = 1,2,\dots r$) where attribute-values $z_i$ of the attributes are obtained from their respective attribute domains $z_i \in d_i$ where $d_i = dom(A_i) \cup \{*\}$. Formally, we shall represent a query involving $r$ queryable attributes by $\bigwedge_{i=1}^{r} z_i$. In this work, we consider only conjunctive combination of attributes as a query. 

On issuance of a query, Twitter API by default returns a result page comprising $m$ (=20) entities and a pointer to the next page of results. The sampler may obtain the subsequent $m$ results for the same query by issuing another API call or get another $m$ results by issuing a different query. Thus, the sampler incurs a unit cost of communication for each query issued. Subsequently, the profiles of entities returned by the API are analyzed to identify whether they satisfy the hidden property of depression or not. Since the cost incurred in determining the entity's hidden property is directly proportional to the API call cost, we use API cost as the sole cost constraint of the problem. 

Now, we formally define the hidden sampling problem as follows. 

\textbf{Problem definition:} \emph{Suppose entities on an online social platform are queryable using a conjunctive combination of $r$ queryable attributes $A_i$ for $i=1, 2, \dots r$. Further, consider that there exists a target subpopulation satisfying a hidden property that is verifiable by an oracle (usually a classifier). Given a budget $B$ of API calls, the sampler's objective is to maximize the count of sampled entities satisfying the hidden (target) property.}

\textbf{Problem hardness:} The sampling of hidden population is challenging in the real-world setting due to several reasons. One, the number of possible queries that can be constructed using the queryable attributes ($\prod_{i} |d_i|$) is exponential in the size of attribute domains. Given a fixed API budget, query selection from an exponential query space makes the problem particularly challenging. Two, the API returns $m$ or fewer results, and when similar queries are used, it often leads to re-sampling of same entities. It therefore becomes pertinent to handle the re-sampling issue so that maximum number of distinct hidden entities can be sampled. Three, the limited API calls necessitates that the sampler manages an exploration-exploitation tradeoff. Given that the sampler is oblivious of the correlation between the queryable attributes and the hidden property, it has to tradeoff the API calls between learning the correlation and issuing the highly correlated queries. In addition, a hidden population sampler is intended to be used for sampling a diverse set of entities such as books, restaurants, products or people through the web interface of online platforms such as the library, Yelp, Amazon and LinkedIn respectively. Given the variety of online platform settings comprising of different sorts of queryable attributes and an unspecified hidden property, an ideal sampler should exhibit the following characteristics. 

\emph{Simplicity}: The model should be applicable to web-forms having different number and types of attributes with varying attribute cardinality. 
\emph{Online}: The model can be updated using only the feedback obtained by the returned result analysis. 
\emph{Unsupervised}: The lack of training examples is usually the prime motivation behind hidden population sampling. Prior distributional information about the hidden population is unavailable to the sampler; thus necessitating the algorithm to be unsupervised. 
\emph{Task-independence}: The sampler is oblivious of the hidden property, and it could therefore be used for sampling any well-defined hidden population, i.e., the sampler should be able to adapt when plugged-in with a different black-box classifier used as the oracle. 
\emph{Perpetual}: Ideally, a sampler is expected to be efficient both when the budget is limited and when the budget is asymptotically large.
\emph{Flexibility}: The model could be easily extended when extra information such as attribute semantics or attribute correlations is partially available.

In the following section, we propose a hidden population sampler that exhibits every aforementioned characteristic. \Cref{tab:terminology} lists some of the frequently used terms and their definitions in this paper.

\begin{table}
    \caption{Terms and their definitions.}
    \label{tab:terminology}
    \begin{tabular}{p{1.5cm}rp{4.5cm}}
        \toprule
        \textbf{Term} & \textbf{Notation} & \multicolumn{1}{c}{\textbf{Definition}}\\ \midrule
        Page size & $m$  & Maximum number of results returned in a single API call. \\
        Query precision & $p_q$  & Fraction of target entities in the population matching a given query.\\
        Attribute cardinality & $|dom(A_i)|$ & Number of attribute-values of a given attribute.\\ 
        Attribute domain & $dom(A_i)$ & Set of all possible attribute-values of a given attribute. \\ 
        \bottomrule             
\end{tabular}
\end{table}

\section{Decision tree-Multi Armed Bandit}
\label{sec:proposed crawler}

In this section, we discuss in detail our proposed \textit{unsupervised online} method for sampling the hidden target population. 

\subsection{Decision problem}
\label{subsec:decision}

We show that the process of sampling hidden population from online social platforms as described in \Cref{sec:problem statement} is primarily a decision problem. The sampler continuously decides which query to issue to the API such that the sampler obtains a maximum possible number of entities from the hidden population within the given API budget. Based on the sampled entities, the sampler maintains a probabilistic model of the entity database that gets updated over time. The model is used to construct a query. The returned-results obtained from issuing the constructed query is subsequently used to update the model. This cycle of query construction, returned-result analysis, and model updation continues until the API budget runs out. We deliberate upon each component of the cyclic process separately. 

We maintain the model of the entity database using a set of probabilistic parameters. In the absence of any prior semantics or syntactic information about the attributes, the model treats each queryable attribute such as location, time and keywords in Twitter as independent variables. Furthermore, we model the attribute-values of every attribute independently, i.e. `New York', `Los Angeles` and `Chicago' corresponding to location attribute is modeled independently as well. The above assumptions concerning the entity database allow our model to be applicable across a suite of online social platforms. The probability model of the entity database is used to not only estimate the utility of issuing each possible query but also to construct the next query.

The sampler interacts with the online social platform via the query interface. As stated in \Cref{sec:problem statement}, we represent the query as a conjunctive combination of discrete attributes. In compliance with our problem formulation, we approximate continuous attributes by discretizing them into different bins and handle text search by an expert-based selection of a few relevant textual phrases. Note, that the number of possible queries that can be constructed using the queryable attribute is still exponential, i.e. $\prod_{i} |d_i|$ which is typically very high. The exponential number of choices makes the decision problem even more challenging.

On issuance of an attributed query, the online social platform returns a list of entities: `returned result set'. As shown in the Twitter example, the same attributed query can be used to gather more results by traversing over the next pages. The returned results act as feedback for the sampler which is used to update the model. The number of entities belonging to the hidden population indicates the quality of the query. Thus, the core objective of the hidden population sampler is to find high-quality queries and to issue those queries repeatedly.

Next, we describe a detailed solution to the cyclic process of decision making for hidden population sampling by employing a decision-tree guided multi-armed bandit algorithm. In the later sections, we show the utility of our proposed sampler over a range of tasks.

\begin{figure*}
 \centering
 \includegraphics[width=\textwidth,trim={0cm 17.5cm 0cm 0cm},clip]{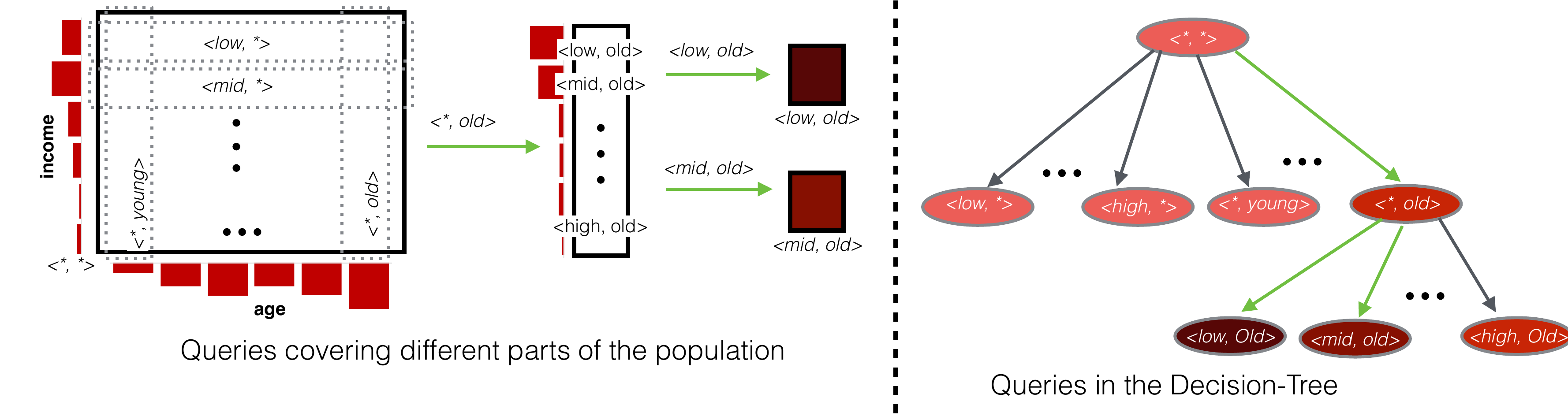}
 \caption{\small Model description of \texttt{DT-TMP}. For hidden population of `mental illness' (represented in red color), the \texttt{DT-TMP} searches the population for the best combinatorial query comprising of two queryable attributes: income and age. It first uses \texttt{<*, *>} query to find the best single attributed query from queries such as \texttt{<Low, *>} and \texttt{<*, Young>}. Subsequently, it finds the best query \texttt{<Low, *>} along which it expands its query search. The decision tree on the right shows the query expansion withe the query expansion along green links. }
 \label{fig:decision_tree_thompson}
\end{figure*}

\subsection{Proposed \texttt{DT-TMP} algorithm}
\label{subsec:proposed algorithm}
In this section, we present solutions to the three prominent challenges encountered by a hidden population sampler in the real-world setting. The challenges are exponential query space, unconventional reward feedback, and an unknown correlation between queryable attributes and hidden property. Next, we fully describe the proposed Decision-Tree Thompson sampler (\texttt{DT-TMP}).

First, as noted in the problem statement, any hidden population sampler constructs its queries by choosing attribute-values $z_i$ from $d_i$ of each queryable attribute $A_i$. The size of query space is therefore exponential in the attribute cardinality $\prod_{i} |d_i|$. We deal with the problem of exponential query space by hierarchically ordering the queries from the most general to the least general (or the most specific) query. \Cref{fig:decision_tree_thompson} shows the hierarchical organization of queries. For instance, a query where location attribute is set to `Chicago' and text attribute is ignored (or set to `*') is a generalization of the query where location is set to `Chicago' and text attribute is set to `\#Cubs' since the former includes all entities matching the later. In principle, a query $q_1$ is a generalization of query $q_2$ if the set of population entities matching query $q_1$ is a superset of the set of entities matching $q_2$. Hence, the most general query is one where $z_i$ set to `*' for every attribute.

Second, the analysis of the returned result set by the API is a non-trivial task because of the partial information available during sampling and the re-sampling issue. In each API call, the sampler obtains partial information in the form of the returned set of $m$ or fewer results. For illustration, consider that a query where location attribute is set to `Chicago' yields 5 entities from the depressed population out of 20 returned entities on the first result page. We shall assume the query precision or the fraction of hidden entities matching this query is $5/20 = 0.4$. In other words, it is very likely to obtain another 5 hidden entities when a new API call is made for the next result page of the same query. More generally, we model probabilistically the query precision using a Beta distribution which is typically used to model probability of probabilities~\cite{robinson2017introduction}. Furthermore, a query where location is set to `Chicago' is very likely to lead to the entities that also satisfy queries where the text attribute is `\#Cubs'. The sampler, therefore, needs to update the quality metric of queries where text attribute is set to `\#Cubs' so that it avoids making redundant queries that lead to re-sampling of same hidden entities. We avoid the re-sampling by estimating the expected number of distinct unseen entities to be discovered by issuing a query $q$. Without making any assumption about the ordering of results for a specific query, we assume that the results are returned either via sampling with or without replacement from the set of entities matching the given query. For sampling with replacement, we derive a reward function as follows.

Assume that the number of entities in the database satisfying a query $q$ is $N_q$. Further, assume that the sampler has already observed $n_q$ distinct entities satisfying the query $q$ out of which there are $S_q$ number of target entities and $F_q$  number of non-target entities. From Standard Probability Theory, we therefore obtain the following expected reward $r_q$ when query $q$ is executed.
\begin{equation}
 \mathbb{E}[r_{q}] = \underbrace{\frac{S_q}{S_q + F_q}}_\text{expected \# targets} \cdot \underbrace{\frac{N_q-n_q}{N_q}}_\text{new} \cdot \underbrace{\Big( 1- \big(1 - \frac{1}{N_q} \big)^{m}}_\text{unique} \Big)
 \label{eq:expected_reward}
\end{equation}

where $m$ is the maximum number of results returned by the web API in response to a single query. The expected reward is the estimated number of new distinct hidden entities that are likely to be obtained by re-issuing the query $q$. 

For sampling with replacement, we update the reward function by dropping the third term since unique entities within the result page of a query is guaranteed (i.e. $\frac{S_q}{S_q + F_q}\frac{N_q-n_q}{N_q}m$). When $N_q$ is unknown, we assume that $N_q >> n_q$, therefore the reward function approximates to just the first term (i.e. $\frac{S_q}{S_q + F_q} m$).

The query precision for any query is an unknown measure to an unsupervised hidden population sampler. If the precision values are known, the sampler would straightforwardly formulate its queries using only the high precision queries. In order to obtain an unbiased estimate of the precision, the sampler needs to explore over the combinatorially large query space. While a naive exploration of queries is beneficial for formulating better future queries learned from the unbiased estimates, it leads to poor immediate results. We, therefore, employ Thompson sampling for handling this exploration-exploitation tradeoff of queries. Thompson sampling is a well-known optimal MAB algorithm that achieves the lower regret bound of the MAB problem~\cite{agrawal2013further}. Notice that Thompson sampling is consistent with the independence assumption of attributes and attribute-values made in \Cref{subsec:decision}.

Now, we generalize the intuitions presented above to propose a simple yet effective hidden population sampler: Decision-Tree Thompson sampler (\texttt{DT-TMP}). 

\underline{Description of algorithm}: \texttt{DT-TMP} is a unique combination of a standard Decision Tree~\cite{quinlan1986induction} and Thompson sampling~\cite{thompson1933likelihood}. \texttt{DT-TMP} maintains a query pool $\mathcal{Q}$ comprising of queries explored by the algorithm. The query pool is initialized with the most general query. Typically, the most general query can be represented using the queryable attributes as $\bigwedge_{i=1}^{r} *$. The most general query is initially used to sample from the entire population. \texttt{DT-TMP} expands the query pool by adding more specific queries. 

For every query $q \in \mathcal{Q}$, \texttt{DT-TMP} tries to predict the future reward that would be obtained when query $q$ is issued. Based on the prediction, \texttt{DT-TMP} chooses the best query to issue. A query issued to the API yields a result page comprising $m$ entities. Each returned entity is evaluated as a success or failure depending on whether it belongs to the hidden population or not. We model each success and failure of every returned entity as a random sample drawn from a unknown Beta distribution of the query that the model learns over iterations. We use a non-informative uniform prior $Beta(1, 1)$ as the starting state for every query. This choice of Beta distribution permits us to efficiently update the posterior distribution upon receiving the returned results. 

We now show how to update the posterior distribution of any query $q' \in \mathcal{Q}$ when another query $q$ is issued. If $q$ is a generalization of $q'$, we increment the success or failure parameter of Beta distribution by one depending on whether the returned entity is in target populace or not, and the returned entity matches query $q$. In the other case, when the specific query $q$ accounts for only a fractional part of the general query $q'$, we update the Beta distribution of the general query proportionately. That is, if $q$ is a specific version of $q'$, we increment the success or failure parameter of $q$ by the ratio of population size matching query $q$ to population size matching query $q'$. We are able to estimate this fraction directly from the returned result since the query pool is expanded hierarchically from the most general to the most specific queries. 

At each step of the iteration, \texttt{DT-TMP} employs Thompson sampling to select the best query among the query pool. Note, that the query pool is fixed over epoch time $h$ to ensure that enough entities are sampled before expanding the query pool. The query pool is expanded by adding new specific queries corresponding to the best query in $\mathcal{Q}$. We prove using Lemma \ref{lemma:centering} that expansion of a general query always leads to an equally good or a higher precision specific query. Thus, \texttt{DT-TMP} continues to find the highest quality query until the budget is finished.  

\begin{lemma}
The query precision of the specific queries are centered around the query precision of their corresponding general query. Further, there exists a leaf node of the decision tree with the highest quality precision.

\label{lemma:centering}
\end{lemma}

\begin{proof}
Lets assume that the query precision of a general query $q_g$ is $p_g$, and it has $n$ immediate specific queries $q_i$ as children in the decision tree. Assume that the query precision of the $i$-th specific query $q_i$ is $p_i$ where $i \in \{ 1, 2, \dots n \}$. Further, we denote $f_i$ as the ratio of the size of the population matching query $q_i$ to the size of population matching the general query $q_g$. Since the disjoint specific queries cover the general query, $\sum_{i} f_i = 1$.

By preservation of hidden entities, the query precision of the general query $p_g$ can be expressed as the weighted average of the specific queries's precision. That is, 
\[ p_g = p_1 f_1 + p_2 f_2 + \dots p_n f_n \]

Since, $p_g$ is a weighted average, it is bounded by the maximum and minimum of the specific queries' precision.

 Following the above argument, we note that there always exists a specific query whose precision is strictly greater or equal to the query precision of the general query. Since, the leaf nodes are the most specific queries in the decision tree of \texttt{DT-TMP}, it therefore follows that one of the leaf nodes of the decision tree has the highest precision. 
\end{proof}

\Cref{algo:decisiontree_thompson} summarizes the \texttt{DT-TMP} algorithm via pseudo-code. We provide a diagrammatic representation of \texttt{DT-TMP} algorithm in ~\Cref{fig:decision_tree_thompson} in a stylized sampling environment where there are only two binary queryable attributes. Next, we perform a detailed analysis of the \texttt{DT-TMP} algorithm.

\begin{figure}
 \centering
 \includegraphics[width=\linewidth]{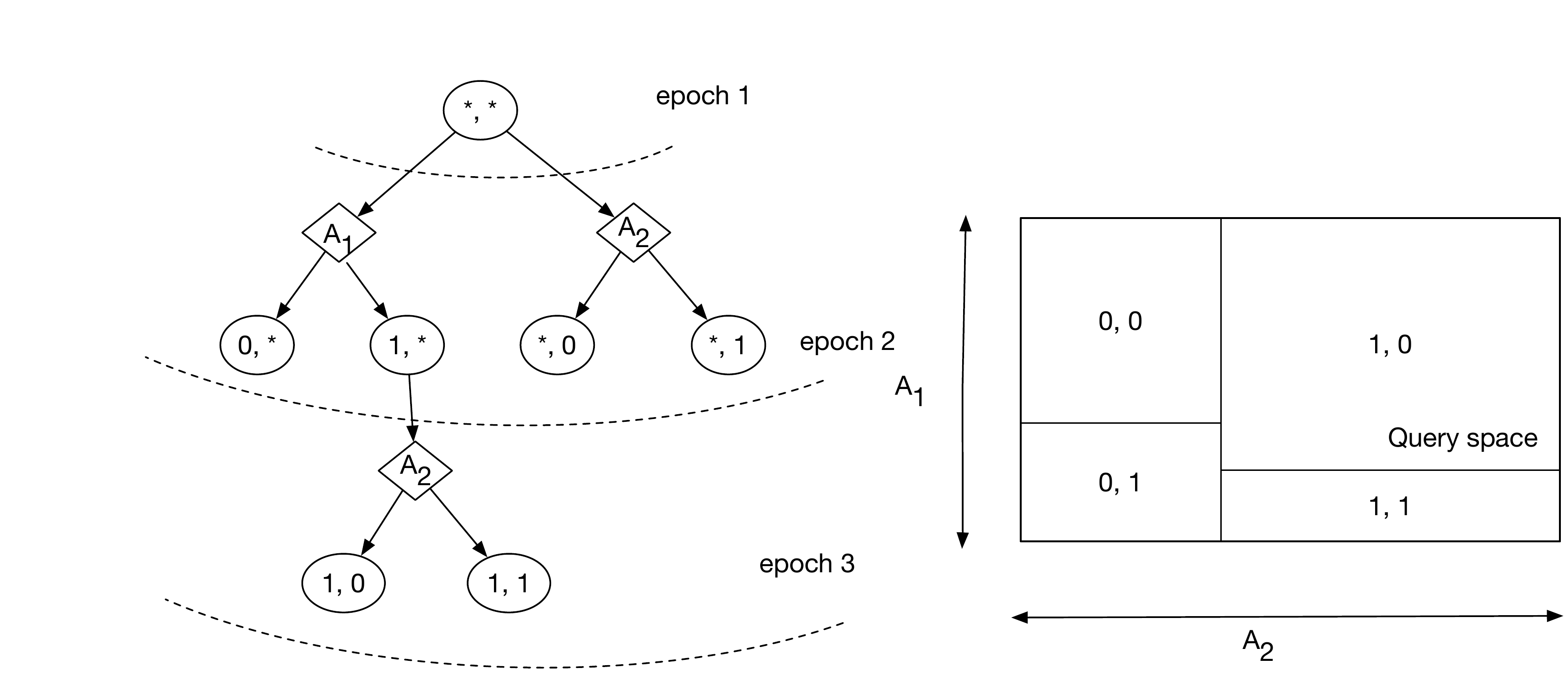}
 \caption{Decision tree diagram shows how new attribute-values pairs are added for exploration as \texttt{DT-TMP} samples queries from two binary attributes $A_1$ and $A_2$. The query at root of the tree covers the entire database represented by box in right subplot. General queries are situtated at higher levels of the tree while specific queries that cover only smaller subsets of population are situated at lower levels of the tree.}
 \label{fig:decision_tree_thompson}
\end{figure}

\begin{algorithm}
 \caption{Decision tree- Thompson sampler}\label{algo:decisiontree_thompson}
 \begin{algorithmic}[1]
  \State $S_q = 1, F_q = 1, \forall q.$ $\vartriangleright$Successes and failures of query $q$
  \State $\mathcal{Q} = \{ \bigwedge_{i=1}^{r} * \}$   \Comment{query pool}
  \State $\mathcal{S} = \phi$  \Comment{sample set of entities}
  \For {$t = 1, 2, \dots B/h$}  \Comment{Communication rounds}
  \For{ $j= 1, 2, \dots, h$ }
  \For {$q \in \mathcal{Q}$}
  \State $r_q \sim$ Beta ($S_q, F_q$) $\times \frac{N_q-n_q}{N_q} \times N_q \Big( 1- \big(1 - \frac{1}{N_q} \big)^{m} \Big)$ 
  \State \Comment{for sampling with replacement}
  \EndFor
  \State $q^{\ast} = argmax_{q \in \mathcal{Q}} \ r_q$
  \State $R$ = Execute query $q^\ast$ \Comment{Returned result}
  \State $\mathcal{S}$ = $\mathcal{S} \cup R$
  \State $s_{q^\ast}$ = Number of successes in $R$
  \State $S_{q^\ast}, F_{q^\ast} = S_{q^\ast}+s_{q^\ast}, F_{q^\ast}+ |R| - s_{q^\ast} $
  
  \State \# update Beta parameters of other queries
  \For {$q' \in descendant(q^\ast) \cap \mathcal{Q}$} 
  \State \Comment{Descendent nodes in DT are specific queries of $q^\ast$}
  \State $S_{q'}$+= \# success in $R$ matching $q'$
  \State $F_{q'}$+= \# failures in $R$ matching $q'$
  \EndFor

  \For {$q' \in ancestor(q^\ast) \cap \mathcal{Q}$}
  \State \Comment{Ancestor nodes in DT are general queries of $q^\ast$}
  \State $\rho$ = Est. population size of $q$ / population size of $q'$
  \State $S_{q'}$+= $ \rho \times$  \# success in $R$
  \State $F_{q'}$+= $ \rho \times$ \# failures in $R$
  \EndFor
  \EndFor
  \State $\mathcal{Q} = \text{Query-expansion}(\mathcal{Q}, q^\ast, \mathcal{S})$ \Comment{adds new queries to the query pool}
  \EndFor
 \end{algorithmic}
\end{algorithm}

\begin{algorithm}
 \caption{Query-expansion($\mathcal{Q}, q^\ast, \mathcal{S}$)}\label{algo:arm_expansion}
 \begin{algorithmic}[1]
  \For{$j= 1, 2, \dots, r$}
  \If{$A_j == *$}
  \For{$ v \in dom(Q_j) \wedge v \in \mathcal{S}_j$} $\vartriangleright$ add all possible attribute-values observed in sample
  \State $\mathcal{Q} = \mathcal{Q} \cup q^\ast(A_j = v)$ $\vartriangleright$ We add new specific query for unexplored attribute-attribute value pair corresponding to the best query $q^*$
  \EndFor
  \EndIf
  \EndFor
  \Return $\mathcal{Q}$
 \end{algorithmic}
\end{algorithm}

\underline{Analysis of algorithm}: Even though \texttt{DT-TMP} models the complex relationship among queryable attributes and hidden attributes while keeping an exploration-exploitation tradeoff among different queries, it is surprisingly easy to implement and has a linear space and quadratic time complexity. For practical reasons, we assume number of attributes $r$ and result size $m$ to be constant. In each iteration in the outer-loop, the maximum number of queries added to the query pool is limited by $n$ where $n= \sum_{i=1}^{r}|dom(A_i)|$; the budget $B$ limits the number of iterations. Furthermore, the decision tree takes $\mathcal{O}(\mathcal{Q})$ space which is bounded by $\mathcal{O}(nB)$. Every query in $\mathcal{Q}$ uses a constant space parameter set to estimate the reward distribution. Thus, the overall space complexity of the \texttt{DT-TMP} is $\mathcal{O}(nB)$. A similar analysis implies the time complexity of \texttt{DT-TMP} is $\mathcal{O}(n^2B)$ since each iteration (sampling from and updating of Beta distributions is performed in constant time) involving updation of specific and general query combinations take $\mathcal{O}(n)$ time. Finally, we notice that \texttt{DT-TMP} can easily be parallelized by having the worker nodes use the preceding iteration's reward distributions while having the master node maintain the central decision tree.

Lastly, we note that at limiting budgets \texttt{DT-TMP} behaves as \texttt{TMP} when the query pool expands to the entire query space.
\begin{lemma}
At asymptotic limits of the budget, \texttt{DT-TMP} tends to a naive Thompson sampler. 
\end{lemma}

\begin{proof}
Given that every branch is initialized with a query precision $Beta(1, 1)$, there is a non-zero probability of selecting any branch using best query selection of \texttt{DT-TMP} (line 9 of \Cref{algo:decisiontree_thompson}). \texttt{DT-TMP} will therefore explores every branch of search tree at asymptotic limits of budget. Since, at asymptotic budget limits, all branches of the decision tree will be explored, hence the query pool will expand to cover the entire query space. When the query pool $\mathcal{Q}$ covers the entire query space, \texttt{DT-TMP} and \texttt{TMP} are identical. 
\end{proof}

\subsection{Guarantees of the proposed algorithm}
\label{subsec:guarantess}

Different number of attributes, attribute cardinalities, attribute distributions and the use of decision-based search tree structure makes the analysis of \texttt{DT-TMP} difficult. Furthermore, it is non-trivial to extend the standard regret analysis used for analyzing MABs to the \texttt{DT-TMP} algorithm. First, unlike MAB which has just one optimal arm (or query), a standard \texttt{DT-TMP}'s optimality involves a set of queries. Second, the underlying quality of a query is fixed in MABs while \texttt{DT-TMP} has unconventional reward feedback as described in~\Cref{subsec:proposed algorithm}. 
Third, on the issuance of a query, the MABs get one result while \texttt{DT-TMP} gets the result set $R$ that can be of size anywhere between $0$ and $m$.

In the following lemma,  we show that when specific query's quality is correlated with the general query's quality, \texttt{DT-TMP} based search tree is useful. 


\begin{lemma}
For a sufficiently large dataset when the query precision of specific queries within one general query are more similar to each another than specific queries of other general query (clustering effect), \texttt{DT-TMP} requires fewer number of queries to find the optimal query than \texttt{TMP}.
\end{lemma}


\begin{proof}
For this proof, we shall consider three scenarios of the clustering of the precision values of specific queries centered around their corresponding general queries. 1) when the clusters are well separated or the general queries are disjoint, 2) when any two cluster among the clusters are disjoint except few specific queries that are common to both. 3) when the clusters corresponding to the general queries are overlapping. 

First, consider that there are $n$ disjoint general queries $q_1, q_2, \dots q_n$ and their corresponding query precision be $p_1, p_2, \dots p_n$. Further, assume that there are $n_i$ specific queries, $q_{i, j}$ where $j \in \{1, 2, \dots n_i\}$ within a general query $q_i$. From Lemma \ref{lemma:centering}, the query precision of specific queries $q_{i, j}$ are clustered around $q_i$'s precision, i.e. the specific query $q_{i,j}$'s precision $p_{i,j}$ lies in the range $[p_i - \Delta p_i, p_i +  \Delta p_i]$ where $\Delta p_i > 0$. Since, the clusters are well separated we note that the query precision of the specific queries satisfy the following condition: $|p_i - p_j| > \Delta p_i + \Delta p_j$ for any pair of general queries $q_i$ and $q_j$. The above condition would imply that the all specific queries corresponding to a general query are well-separated. First, we shall proof the separation of clusters composed of precision values of specific queries under the aforementioned condition. Next, we shall proof the efficacy of \texttt{DT-TMP} over \texttt{TMP} when the clusters satisfy the condition.

We shall now prove that the precision range of specific queries corresponding a general query $q_i$ is disjoint from the range of specific queries' precision corresponding to another general query $q_j$ where $i \ne j$. 

Without loss of generality, we assume that $p_i > p_j$. From the clustering condition stated above, $ p_i - p_j > \Delta p_i + \Delta p_j$. By re-arranging the terms, we get 
\[p_i - \Delta p_i > p_j + \Delta p_j\]
Since, $\Delta p_i, \Delta p_j > 0$, thus 
\[p_i+\Delta p_i > p_i - \Delta p_i >  p_j + \Delta p_i > p_j - \Delta p_j\]
It follows form the above inequalities that,  \[[p_i-\Delta p_i , p_i+\Delta p_i] \cap [p_j-\Delta p_j , p_j+\Delta p_j] = \phi\]

Thus, we observe that the precision range of specific queries belonging to a general query $q_i$ which is $[p_i-\Delta p_i , p_i+\Delta p_i]$ doesn't overlap with precision range induced by any another general query $q_j$.

We shall now use the above clustering condition of the precision ranges of the specific queries to show that the \texttt{DT-TMP} requires lesser number of samples than \texttt{TMP} to find the best query in the query space. 

Without loss of generality, assume that $p_1 > p_2 > \dots > p_n$. Further, assume the query precision of a specific query $q_{i, j}$ within the general query $q_i$ are ordered by their precision values $p_{i, j}$. We shall now show the sample complexity of \texttt{DT-TMP} for identifying the best query $p_{1, 1}$ is lesser than the sample complexity of \texttt{TMP}. Note that it follows from Lemma \ref{lemma:centering} that $p_{1, 1} \ge p_1$.

From sampling theorem~\cite{bousquet2003introduction}, we know that the number of samples from sub-optimal query $q_{i, j}$ needed to discern the best query $q_{1, 1}$ with a confidence interval of $\delta$ is $\mathcal{O} (\frac{1}{\epsilon^2} \ln \frac{2}{\delta})$, where $\epsilon = p_{1, 1} - p_{i, j}$.

\texttt{TMP} searches among the specific queries corresponding to all general queries to identify the best query $q_{1, 1}$. Thus, \texttt{TMP} would have a sample complexity of $\sum_{i=1}^{n} \sum_{j =1}^{n_i} \mathcal{O} (\frac{1}{\epsilon_{i, j}^2} \ln \frac{2}{\delta})$, where $[i, j] \ne [1, 1]$ to find the best query $q_{1, 1}$. 

On the other hand, \texttt{DT-TMP} is a two phase sampler. In the first phase, it identifies the best cluster or general query. In the second phase, it identifies the best specific query within the best chosen general query $q_1$ (cluster). In the second phase, \texttt{DT-TMP} explores among only the specific queries corresponding to $q_1$ to identify $q_{1,1}$. Thus, the sample complexity of \texttt{DT-TMP} to identify the best specific query $q_{1, 1}$ is, 
\[ \underbrace{ \mathcal{O} (\sum_{i=2}^{n} \frac{1}{(p_1 - p_i)^2} \ln \frac{2}{\delta}) }_\text{phase 1} +  \underbrace{ (\sum_{j=2}^{n_1} \frac{1}{(p_{1, 1} - p_{1, j})^2} \ln \frac{2}{\delta}) }_\text{phase 2} \]

Note, that the phase two of \texttt{DT-TMP} matches with the best query finding among the specific queries of query $q_1$. However, we note that the sample complexity for searching among the specific queries $q_{i, j}$ where $i\ne 1$ is lesser than the sample complexity involved in finding the best clusters. Assuming that $\Delta p_i$'s are similar (i.e., $\Delta p_i \approx \Delta p_j, \forall i, j$), we note that \texttt{TMP} would require around $\mathcal{O}(\sum_{j=1}^{n} n_j \frac{1}{(p_1 - p_i)^2} \ln \frac{2}{\delta}))$ sample complexity in-comparison to \texttt{DT-TMP}'s $\mathcal{O}(\frac{n_1 n}{(p_1 - p_i)^2} \ln \frac{2}{\delta})$.

Thus, we observe that as the number of queryable attributes $n$ increases and the attribute cardinalities of the attributes $n_i$ increases, \texttt{DT-TMP}'s relative improvement over \texttt{TMP} increases.

Second, we note that when the specific clusters are shared, it can be reduced to disjoint case by considering two sub-cases: a) when $q_{1, 1}$ is a shared specific query, thus even if \texttt{DT-TMP} chooses a a sub-optimal query will lead to $q_{1, 1}$. b) when $q_{1, 1}$ is not shared, it can still be proved that the query precision of $q_1$ is highest and thus \texttt{DT-TMP} finds the optimal query. Third, when the clusters are overlapping, it is harder to proof the efficacy of \texttt{DT-TMP} since the best specific query may belong to a sub-optimal general query. We leave the proof of overlapping case for future work.

\end{proof}

\section{Offline Experiments}
\label{sec:experiments_offline}

In this section, we present experimental findings by executing different hidden population sampling strategies over offline real-world datasets. First, we discuss the the real-world datasets and query interface that allows samplers to access the entities within these datasets. Next, we discuss the evaluation metric used to compare the efficacy of the baseline and proposed samplers. Thereafter, we present a brief summary of various baseline sampling strategies. Finally, we show the results of the samplers' performances and interpret the results.

\subsection{Datasets}
\label{subsec:datasets}

Now, we present a summary of three real-world datasets used for offline experiments along with a description of the query interface (API).

In similarity to the offline experiments in \citep{sheng2012optimal}, we simulate a typical online social platform using data from three real-world entity datasets. The datasets---Patent~\cite{hall2001nber}, Auto\footnote{\url{https://www.kaggle.com/orgesleka/used-cars-database}}, and Adult~\cite{kohavi1996scaling}---are obtained from notably diverse domains. The datasets vary in the number of attributes, their attributes' cardinality and the distribution of arm sizes of the attributes (i.e. number of records across attribute values). The local server allows us to vary the parameters of sampling such as the result size, the attribute cardinalities and the attribute correlation values. We query the local server using queryable attributes and evaluate our algorithms based on the coverage of the hidden population entities. \Cref{tab:dataset_stats} summarizes the statistics of the three datasets concerning the respective hidden target population as: patents authored by Japanese researchers, automobiles that have less than $40K$ kilometers mileage and adults who earn more than \$50K per annum. We use all queryable attributes specified in the table for searching the corresponding hidden target population in the datasets.

\textit{Hidden target selection.} For experimental validation, we choose the hidden property of the hidden population as the attributes that are not expressible as a combination of one or more queryable attributes. For example, the queryable attributes such as category, subcategory and class of a patent cannot be used to search for authors with a specific nationality (hidden property). More importantly, the choice of hidden property for offline datasets was motivated by real-world scenarios. We note that academic search engines such as PubMed, Google Scholar and Microsoft Search are not searchable using properties such as the author's nationality and their writing style. Similarly, mileage information is a non-queryable property of the automobiles in the popular advertisement website Craigslist. Lastly, income is often a hidden non-queryable property of users in existing search interfaces of popular sites such as LinkedIn and Angelist but can be easily inferred given their job description. In out datasets, the target population comprises 18.73\%, 5.58\% and 23.93\% of the population size in Patent, Auto and Adult datasets respectively.

\textit{Patent} dataset is a collection of two million patent records from US patent records. The queryable attributes includes categorical attributes such as category and subcategory of patents. Given that nationality of the inventor is a hidden property, we choose the patents that are invented by Japanese researchers as the target entities in the dataset. Patents authored by Japanese researchers was chosen for only illustrative reasons. We observe similar empirical performance across samplers when the target population is set to patents authored by researchers from other countries such as China and India.

\textit{Auto} dataset is a collection of 371,469 auto-vehicle records crawled from Ebay-Kleinanzeigen. The set of queryable attributes allows users to search for cars by searching over a variety of attributes such as car type, model and their brand. For this dataset, we define the hidden target population of automobiles by setting an arbitrary threshold of their travel miles.

\textit{Adult} dataset is a collection of 48,842 records of adults extracted by Barry Becker from the 1994 Census database. Similar to online social networks such as Facebook and Pokec, the users can be searched within the entity database based on their public attributes such as education, marital status and gender. For this dataset, we consider income as the private or hidden attribute of an individual. 

In the absence of a well-defined query interface system, we simulate the query interface in the following way. The query interface to the datasets allows only conjunctive queries formed using the queryable attributes listed in \Cref{tab:dataset_stats}. For each query, the user incurs a unit cost in API call. The query interface returns a set of $k$ random results (default value of page size is set to 10 unless otherwise stated) drawn from the query matching entities with replacement. We observe similar empirical result across samplers when the query interface returned  $k$ random results drawn from the matching entities without replacement.

\begin{table}[h]
	\caption{Description of real-world dataset: Patent, Auto and Adult and their respective queryable attributes with cardinalities in parenthesis, and the target hidden attribute.}
    \label{tab:dataset_stats} 
	\begin{tabular}{cp{4cm}c}
		\toprule
		\textbf{Dataset} & \multicolumn{1}{c}{\textbf{Queryable attributes}} & \multicolumn{1}{c}{\textbf{Hidden property}} \\ \midrule
         Patent  & category (6), subcategory (26), assignee type (7), nclass (417) &  inventor's nationality\\
         Auto & vehicle type (9), model (252), brand (40), fuel type (8), repairing (3) & mileage \\
          Adult & class (9), education (16), marital status (7), occupation (15), relationship (6), sex (2) & income \\
        \bottomrule             
\end{tabular}
\end{table}





\subsection{Evaluation}
\label{subsec:Evaluation}

 We assess the performance of a sampler by the number of distinct entities of the hidden target population that the sampler collects within a given query budget $B$. Alternative definitions include coverage, query harvest rate and precision~\cite{larson2010introduction}. Owing to the similarity in performance of samplers across the aforementioned evaluation metrics, we focus on the simplest evaluation metric: recall. Formally, recall ($R$) of a sampler after making budget $B$ API calls is calculated as follows, 
\begin{equation}
R = \frac{\text{\#(hidden target entities retrieved)}}{\text{\#(target entities in dataset)}} = \frac{N_S}{N_D}
\end{equation}

We observe that recall $R$ is a very small quantity under limited budget constraint because of the large number of $N_D$ target entities in datasets like Patent. Since, it is possible to sample only a limited  $N_S$ number of target entities to be sampled within a limited budget, the overall recall value across all samplers is pretty low ($\frac{N_S}{N_D} << 1$). Therefore, we normalize the recall values by the theoretically maximum recall attainable at a budget $B$ which is same as getting all target entities in every result of API call i.e. $(\text{page-size} \times B)$ unique target entities. 

All evaluation results are reported over 100 independent runs. 

\begin{figure}[h]
\centering
\includegraphics[width=\linewidth, trim={4cm 0.3cm 5.5cm 1cm},clip]{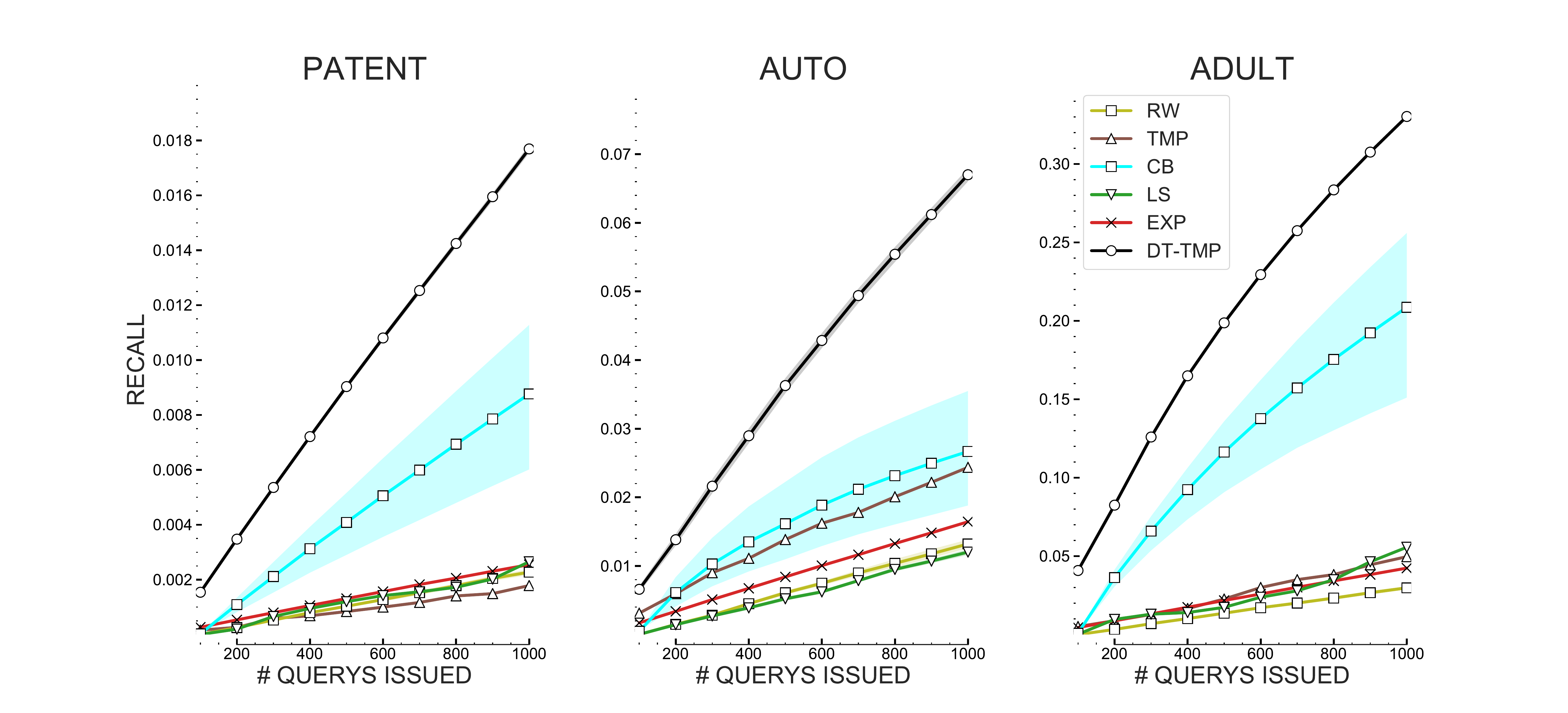}
\caption{Sampling performance of baseline and proposed \texttt{DT-TMP} sampler. \texttt{DT-TMP} is shown to be the best sampling strategy. Decision tree based search allows \texttt{DT-TMP} to explore high yielding queries in a combinatorial query space while simultaneously exploring-exploiting high yielding queries. The bands indicate 95\% two-sided confidence interval. 
}
\label{fig:samplers_performance}
\end{figure}

\subsection{Baselines}

We now enumerate different sampling strategies to compare against our proposed sample.

\begin{itemize}
  \item Uniform sampling over entire database (\texttt{UNI}). As evident from the name, this sampler samples the hidden population by repeatedly issuing the most generic query until it exhausts the budget $B$. Thus, the sample obtained using this sampler is a uniform sample over the entire population. 
  \item Uniform query sampling from the query space or pure exploration sampling (\texttt{EXP}). At each time step, EXP queries the web API by randomly sampling with replacement a single query from all possible queries. Thus, this method performs exploration for $B$ rounds, and hence named as exploration sampling or EXP.
  \item Thompson sampling (\texttt{TMP}) is a standard Thompson sampler~\cite{thompson1933likelihood} where the reward from each arm and draws the best expected arm in each draw. In other words, it is  a \texttt{DT-TMP} sampler without decision tree.
    
\item Lazy slice cover search (\texttt{LS})~\cite{sheng2012optimal} is an optimal algorithm for retrieving the entire entity set from the online social platforms while minimizing the number of queries. 

\item Content-based search (\texttt{CB})~\cite{nazi2015querying} is a greedy query design algorithm that uses the sampled entities to construct new high yielding queries that are unique to the hidden population using the \textit{tf-idf} (term frequency inverse document frequency) ranking.

\item Random Walk (\texttt{RW})~\cite{dasgupta2007random} is an efficient algorithm for randomly sampling from the entity database by creating queries via random walk approach over the space of
queries. 

\end{itemize}

Several works in reinforcement learning literature~\cite{kaelbling1996reinforcement} exist such as Successive Elimination, UCB, UCT which along-with Thompson sampler are known to be optimal for handling exploration-exploitation in MABs, but we exclude them from this study since it is not the focus of this work. 


Finally, we set the epoch $h$ of \texttt{DT-TMP} sampler by default to 10 for all datasets. This setting ensures that the sampler uses feedback gained from $10m$ (typically 100) new observations to expand the query pool appropriately.



\subsection{Results on offline datasets}
\label{subsec:offline_results}
We now present a detailed description of experimental findings obtained from experiments performed on real-world datasets.

Now, we present the performance of baseline and proposed samplers under a variable query budget (ranging from 100 to 1K API calls) on the three real-world datasets---Patent, Auto and Adult. \Cref{fig:samplers_performance} shows the recall value for different sampling strategies across varying query budgets. We observe similar performances of naive MAB based samplers, \texttt{EXP} and \texttt{TMP} along-with \texttt{UNI}. The low performance of naive MAB based samplers is expected given the exponential query space. The naive samplers have to explore among 2600, 11314 and 5970 non-empty queries (queries that have at least one matching entity in the dataset) in Patent, Auto and Adult dataset respectively. Thus, the exponential size of the query space causes these samplers to get stuck in the exploration phase. Notice that there are 630K, 2M, 181K possible queries obtained from the combinatorial combination of attribute-values corresponding to the queryable attribute ($\prod_{i=1}^{r} d_i$) but only a few non-empty queries. In our experiments, we allow the non-empty queries to be used as the arms of baseline samplers; \texttt{DT-TMP} however lacks this information. However, we assume that \texttt{DT-TMP} obtains $N_q$ number of matching entities corresponding to a query $q$ when query is issued as observed in real-world APIs like Google, Amazon, Facebook and LinkedIn API. The \texttt{UNI} sampler that samples random entities from the entire population performs significantly better than naive MAB samplers; however, note that \texttt{UNI} is used as a theoretical baseline and not a feasible sampler for several real-world online platforms like Facebook, LinkedIn and Twitter. Similar to \texttt{UNI}, \texttt{RW} and \texttt{LS} samplers are not target specific samplers. The greedy based \texttt{CB} suffers from the problem of getting stuck in local high quality queries, whereas \texttt{DT-TMP} by the virtue of Thompson sampling maintains an optimal exploration-exploitation tradeoff. Decision tree based MAB sampler (\texttt{DT-TMP}) significantly outperforms the second best baseline sampler by a margin of $101.86\%$, $151.16\%$ and $58.30\%$ in Patent, Auto and Adult dataset respectively. High recall performance of \texttt{DT-TMP} is due to the fact that it exploits the hierarchical structure of the combinatorial action space to greedily explore the attribute combinations (query) that yield a high number of hidden target entities.

In this section, we discussed the experimental outcomes of offline real-world experiments. We used state-of-the-art MAB sampler (\texttt{TMP}) and hidden database samplers (\texttt{RW}, \texttt{LS} and \texttt{CB}) as the baseline samplers. \texttt{DT-TMP} is shown to perform remarkably well over all hidden population tasks across three different datasets by exploiting the structure in combinatorial query space. In the next section, we observe the sampling performances on online real-world datasets. 
\section{Online Experiments}
\label{sec:experiments_online}

In this section, we deploy our different sampling strategies on three real-world online web-query platforms---Twitter, RateMDs and GitHub.

\subsection{Twitter}

Twitter is a popular micro-blogging website that allows users to interact with each other via short messages called as ``tweets''. As of December 2017, Twitter reported an estimate of 330 million monthly active users generate half a billion tweet everyday \cite{twitter2018}. Given the enormous content size and the API limitations, it is infeasible to sample the entire content. Thus, we need to design effective sampling strategies to sieve just the relevant information relating to the hidden population.

Twitter allows combinatorial queries for very few attributes; we use location and hashtags as the two queryable attributes. We consider the hashtags of top 10 National Football League (NFL) teams in USA \footnote{\url{https://www.usatoday.com/sports/nfl/rankings/}} as the first queryable attribute. The second queryable attribute is the home cities corresponding to the 10 NFL teams. 

Twitter REST API is used to gather tweets corresponding to all possible combinatorial attributes (query). Since Twitter API allows only seven days of data to be accessed, we collected all possible queryable tweets between a fixed time frame of 6 days (26 April 2018 till May 1, 2018). Note, that Twitter API \footnote{\url{https://developer.twitter.com/en/docs}} returns by default of 20 tweets per API call and a maximum of 100 tweets. We vary the result-size from 10 to 100 and observe negligible impact of the page size on relative performance of the samplers. Further, the API returns tweets ordered by tweet's creation time, i.e. for a fixed query, Twitter returns the most recent tweets as the first page and older tweets as the subsequent pages.

We consider three hidden population sampling tasks on Twitter. We employ properties of Twitter users that are not queryable via Twitter API to define three hidden populations: female users, users who have verified Twitter accounts, and early adopters of Twitter based on the when the users created their account. We use an off-the-shelf gender predictor~\cite{mullen2015gender} as our ground truth classifier for predicting user's gender; other two properties are available from user profile information. Consider a sports advertiser interested in reaching out to potential football fans who are female on Twitter. Such an advertiser would want to maximize the coverage of female Twitter users in their sample.

We use throughput-rate as the online sampling evaluation metric. In the absence of information about the size of the underlying target population, it is not possible to use recall as the evaluation metric. In the context of hidden population sampling, we define the throughput rate as the ratio of the number of unique target entities sampled to theoretical maximum possible target entities that can be sampled. Thus, the throughput rate ($TR$) of a sampler after making budget $B$ API calls is defined as, 
\begin{equation}
TR = \frac{\text{\#(hidden target entities retrieved)}}{B \times k}
\end{equation}
where $k$ is the page size. Note that throughput rate penalizes queries that yield results of size less page size. For example, consider for page size $k = 20$, a query $q_1$ that returns just a set of 5 target entities has the same throughput rate as another query $q_2$ that returns 5 target and 15 non-target entities, since the API cost incurred by both queries is same, and they both discover 5 target entities.

\begin{figure}
\centering
\includegraphics[width=\linewidth,trim={2cm 0.3cm 5.5cm 1cm},clip]{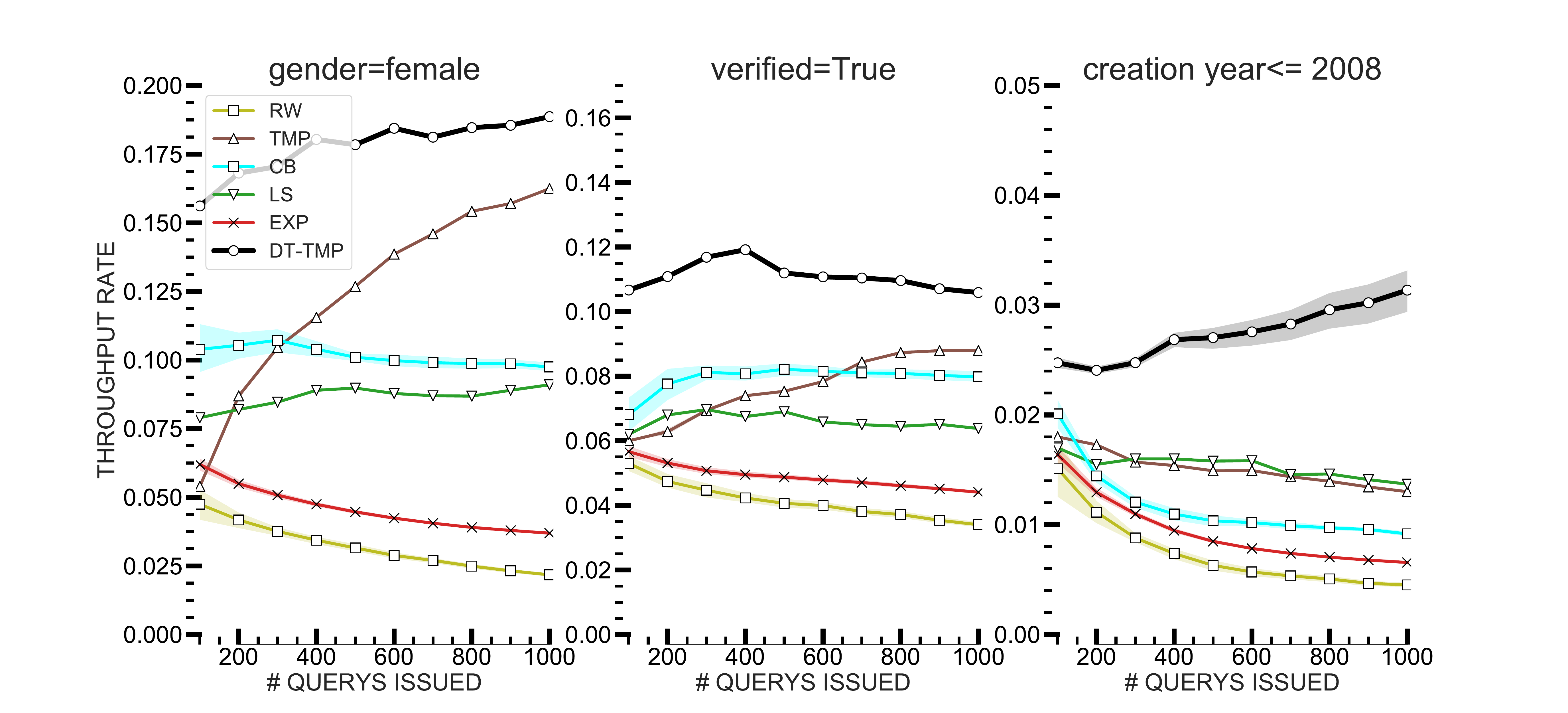}

\caption{Throughput rate for different sampling strategies in Twitter. Combinatorial MAB (\texttt{DT-TMP}) does the best. }
\label{fig:twitter_throughput}
\end{figure}

Due to API restrictions, we can implement only certain samplers on Twitter. Since Twitter API doesn't support uniform sampling, UNI could not be implemented. Therefore, we compare our sampler with five baseline sampling strategies---\texttt{EXP}, \texttt{LS}, \texttt{RW}, \texttt{CB} and \texttt{TMP}. Further, non-random ranking and absence of query size restricts the implementation of \texttt{DT-TMP}. This lack of information about matching results in Twitter of a given query led us to modify the expected reward to the fraction of target entities discovered in each API call. Furthermore, note that the individual attributes, i.e. hashtags and locations are the most general queries on Twitter.

\Cref{fig:twitter_throughput} reveals that \texttt{DT-TMP} is the predominant sampling strategy. It outperforms the second best sampler \texttt{TMP} by a margin of 42.7\% when the task is to sample female users. It shows similar improvement of 39.83\% and 79.24\% over the second best samplers when the hidden target population is set to users that have verified accounts and when the target population of early Twitter adopters respectively. Consistently superior of \texttt{DT-TMP} across different tasks demonstrates the usefulness of the algorithm.

\subsection{RateMDs}
RateMD (\url{https://ratemds.com}) is a free healthcare website that allows users to read and submit reviews about doctors. According to the website, more than 100 million potential patients use RateMDs for information before making important healthcare decisions. There are four queryable attributes---gender, specialties, verified and patient acceptance---used to search for doctors. For gender information, the users have two options, male and female. For specialties, the users can specify one of the 57 specialties of doctors including dentist, pathologist and pediatrician. For patient acceptance, the users can restrict the result on doctors who currently accept new patients. For doctor verification, the users can restrict the search to only doctors verified by RateMD. Each query returns one page of 10 doctors and the user can also specify which page to retrieve. We obtained the dataset by querying 10 pages for each of the attribute-value combinations in August 2018.
Each doctor page is a profile consist of user ratings, credentials and acceptable insurance. 

We consider three hidden population of doctors: doctors with 5-star rating, doctors who received more than 10 ratings, and doctors who accept at least three insurance. The experiments are conducted on the baseline and the proposed samplers; the results are evaluated using the same metric as aforementioned online experiments.  We observe the superior performance of \texttt{DT-TMP} in \Cref{fig:ratemd_throughput} across different tasks. It outperforms the competition by a margin of 55.8\%, 64\% and 25.7\% over the three different tasks. Overall the throughput rate falls with sampling budget. Since, the number of doctors are limited, the rate at which newer target doctors are found decreases over time. 



\begin{figure}
\centering
\includegraphics[width=\linewidth,trim={2cm 0.3cm 5.5cm 1cm},clip]{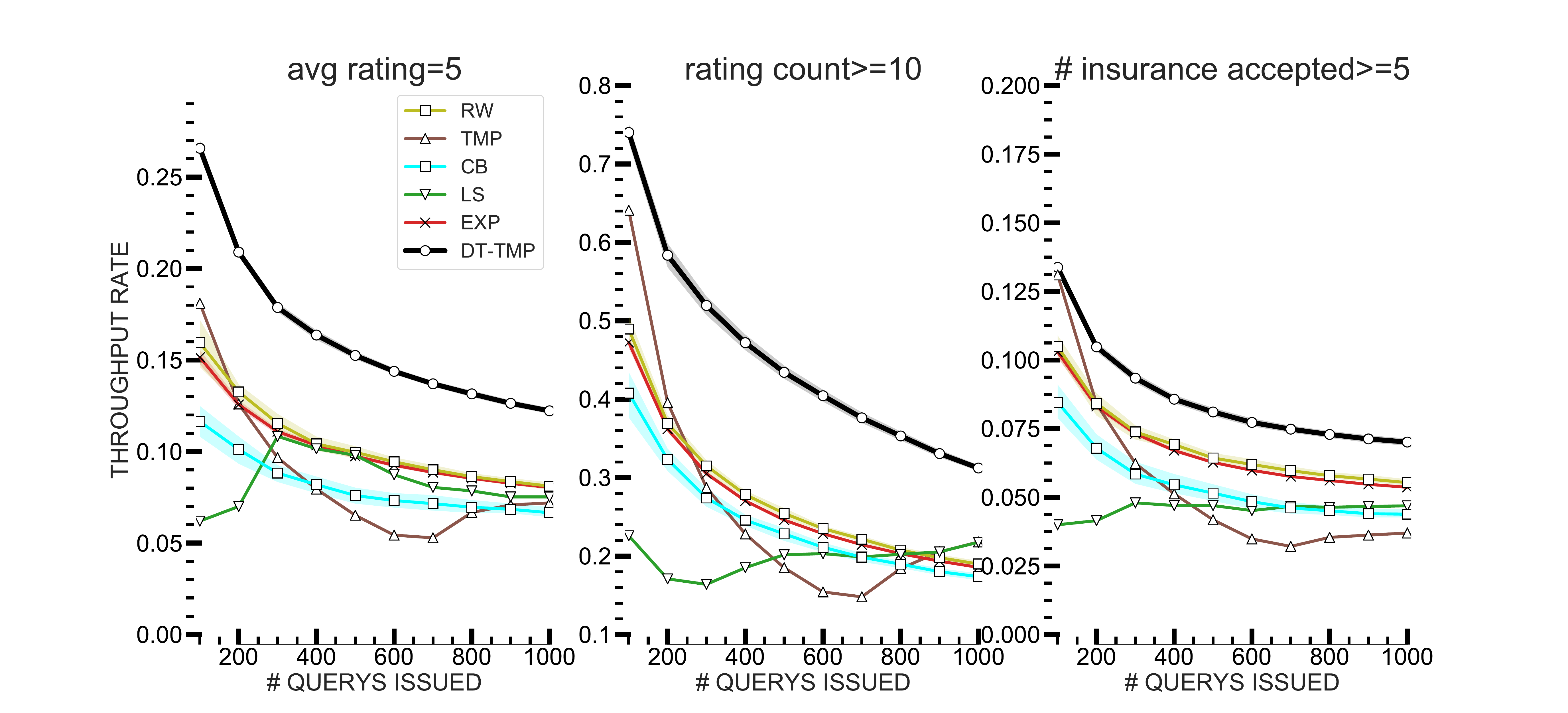}
\caption{Throughput rate for different sampling strategies in RateMD. Combinatorial MAB (\texttt{DT-TMP}) does the best. }
\label{fig:ratemd_throughput}
\end{figure}

\subsection{GitHub}

GitHub is the most popular open-source version control system that allows individuals to manage and collaborate on software-related projects. As of April 2017, Github reported an estimate of 20 million users and 57 million repositories \cite{github2017}. Unlike Twitter, GitHub allows a number of user attributes to query for users. We use three queryable attributes. For the first queryable attribute, we use the ten most popular programming languages used on GitHub \footnote{\url{http://githut.info/}} to search for users using these languages in their projects. For the second queryable attribute, we discretize the ``number of followers'' attribute into three queryable ranges: $\le10$, $>100$ and otherwise. For the third queryable attribute, we discretize the number of repositories of a user into two queryable ranges: $\le10$ and $>10$. Similar to Twitter, Github returns by default a set of 20 users for each query and 100 users at maximum. Furthermore, for a given query, GitHub API returns a maximum possible results of 1000. We use the default ranking of API to get the results of a query.

For the first task, we consider GitHub users whose nationality is China as the first hidden target population. For identifying Chinese users, we employ the location information in the users' profile to predict their nationality. For the next two tasks, we set hidden target population to committed GitHub users (users who contributed to projects on more than 50\% of days in the year 2017) and users who work or study at educational institutions (inferred by their email address). The experiments are conducted on the baseline and the proposed samplers; the results are evaluated using the same metric as aforementioned online experiments.

\begin{figure}
\centering
\includegraphics[width=\linewidth,trim={2cm 0.3cm 5.5cm 1cm},clip]{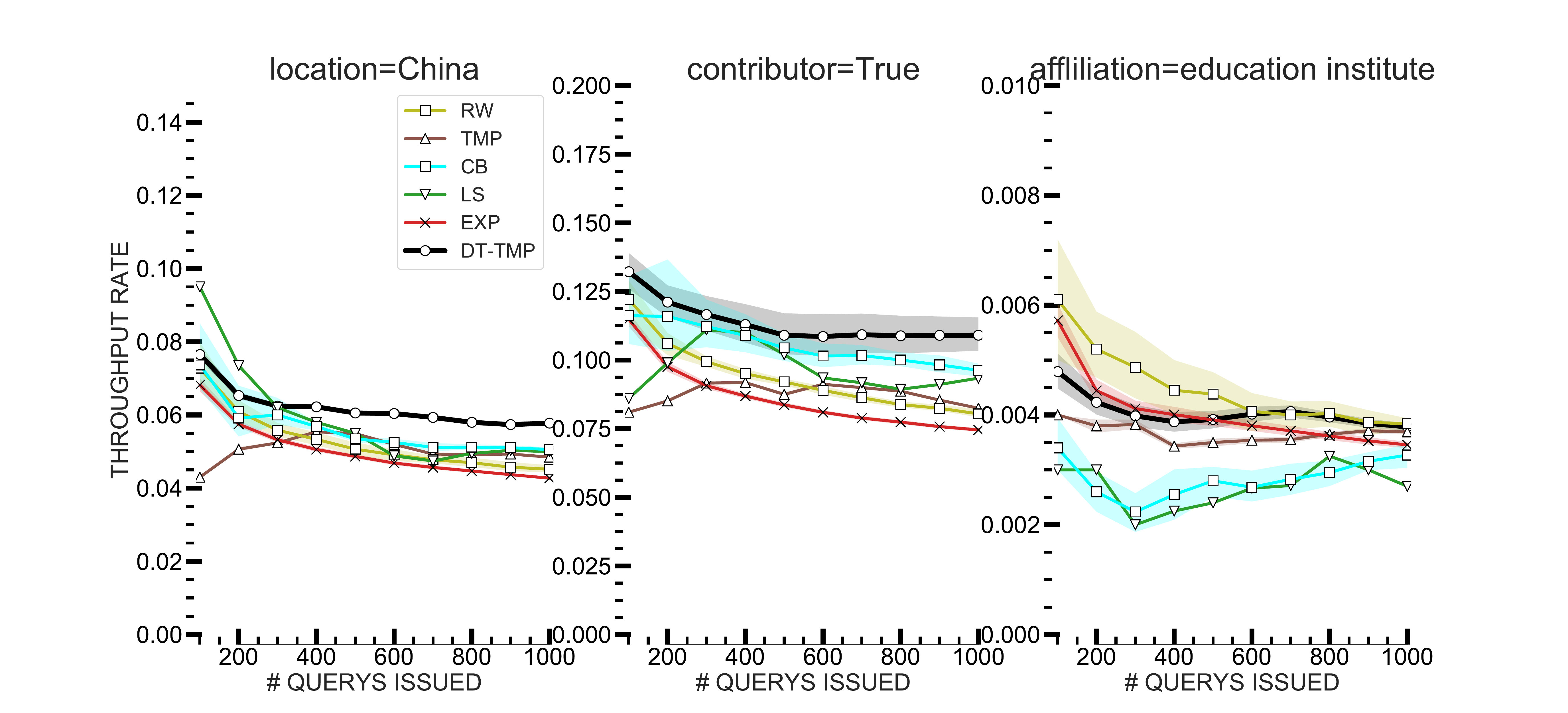}
\caption{Throughput rate for different sampling strategies in Github. Combinatorial MAB (\texttt{DT-TMP}) does the best. }
\label{fig:github_throughput}
\end{figure}

 \Cref{fig:github_throughput} reveals that \texttt{DT-TMP} is the best sampling strategy. It is statistically the best sampler for the first two tasks at a confidence interval of $95\%$. For the last task, since the overall number of accounts declaring education affiliation is very low, the feedback is very weak leading to similar poor performance across all task. 
\section{Discussion}
\label{sec:discussion}

In this section, we discuss ways to improve sampling of hidden target population by using attribute combinations for querying. In ~\cref{subsec:limitation}, we discuss limitations of our proposed approach.

\subsection{Why does combinatorial querying work?}
\label{subsec:independent_attr_results}

We observe that querying online APIs via combination of one or more attributes leads to higher coverage of hidden target population than querying via individual attribute at a time. In a non-combinatorial querying system, only one attribute can be used to define a query. Therefore, each queryable attribute $A_i$ contributes $d_i$ ($A_i$'s cardinality) different queries to the non-combinatorial querying space. Whereas, a combinatorial query space is defined by conjunction of one or more queryable attributes. The size of combinatorial query space is exponential. One of the advantages of non-combinatorial querying system over combinatorial querying system is its limited size of query space. This facilitates existing reinforcement learners to efficiently explore-exploit high yielding queries in non-combinatorial query systems. However, we shall show via \texttt{DT-TMP} sampler that correlation between attributes and hierarchical structure within combinatorial query space can be exploited to design even more efficient sampling strategies.

Furthermore, using a specific sampling scenario, we prove the advantage of using combinatorial query system over non-combinatorial query system.

\begin{lemma}
For a sufficiently large dataset and a query sytem defined over a uniformly distributed attribute, \texttt{DT-TMP} requires a fewer number of queries to find the best query when the universal query (`*') is available than \texttt{TMP}. 
\end{lemma}

\begin{proof}
Lets consider a query system defined over a uniformly distributed attribute with $1, 2, \dots k$ attribute values. Further, consider that $m$ samples are returned in a single result page for each API call. Further assume that only $m'$ samples from each attribute value is needed to discern the best attribute value at a given confidence interval $\delta$~\cite{even2002pac}. 

Thus, a naive query system would require $\lceil \frac{m'}{m} \rceil$ API calls corresponding to each attribute value to figure out the best query. Thus the total number of API calls issued by the non-combinatorial query system is $k \lceil \frac{m'}{m} \rceil$. 

However, consider another query system that allows `*' query. In other words for each API call, we obtain $m$ samples that are drawn from any of the $k$ attribute values. Given that the attribute is uniformly distributed, therefore only $\lceil k \frac{m'}{m} \rceil$ expected number of API calls are required to obtain $m'$ samples for each attribute value. 

Since, $\lceil k \frac{m'}{m} \rceil < k \lceil \frac{m'}{m} \rceil $, therefore a combinatorial query requires fewer number of queries in expectation than a non-combinatorial query system.  
\end{proof}

%
%
%

\Cref{fig:recall_independent_vs_combination} shows our proposed \texttt{DT-TMP} sampler that uses combinatorial querying system outperforms all non-combinatorial based samplers. Owing to the efficient utilization of hierarchy in query space, \texttt{DT-TMP} outperforms its competition. At a query budget of 1000, \texttt{DT-TMP} outperforms the best non-combinatorial query sampler, \texttt{TMP}, by margin of 112.75\%, 23.25\% and 10.64\% on the datasets---Patent, Auto and Adult respectively.

\begin{figure}
 \centering
 \includegraphics[width=\linewidth,trim={5cm 0.3cm 5.5cm 1cm},clip]{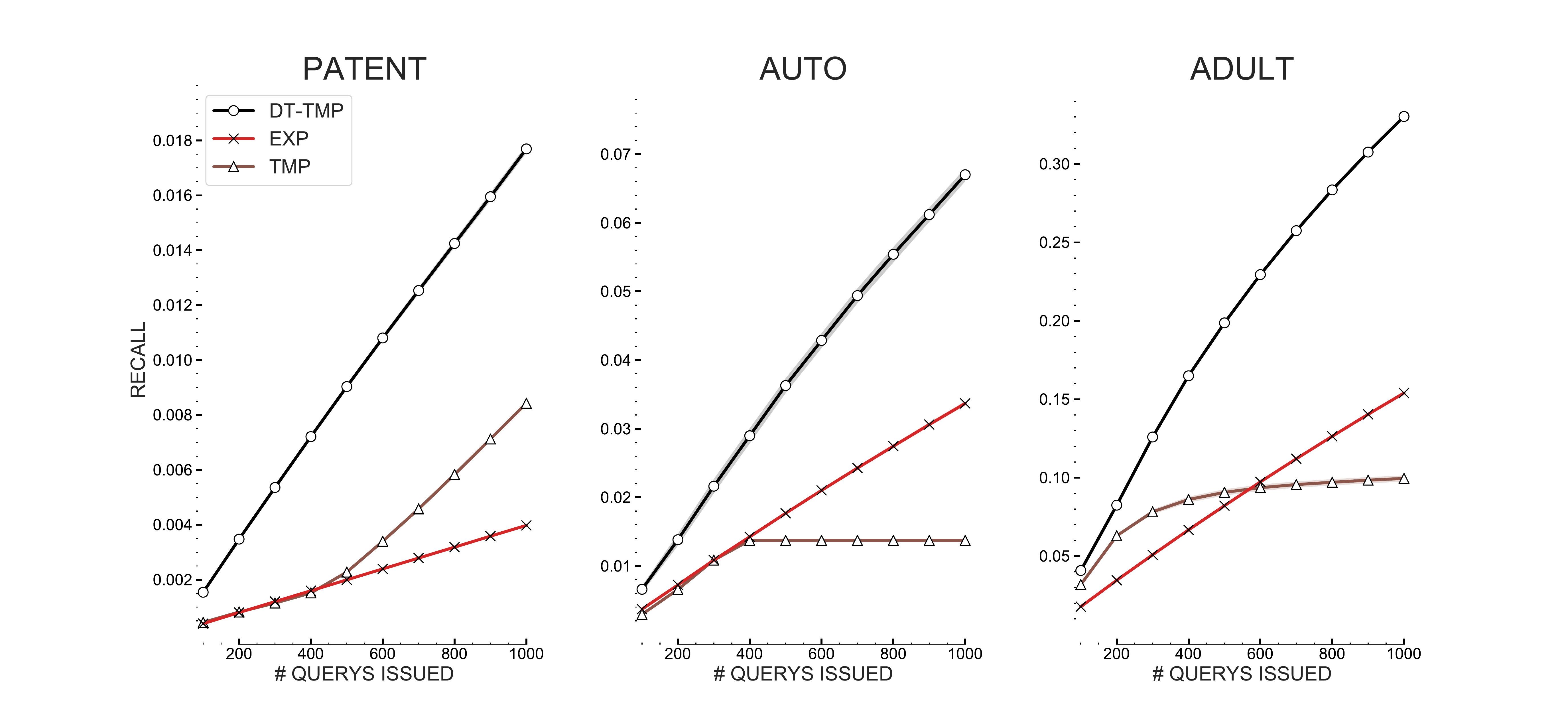}
 \caption{Comparing the recall value of non-combinatorial query based baseline samplers (\texttt{TMP}, \texttt{EXP}) and combinatorial sampler (\texttt{DT-TMP}) in real-world datasets. Combinatorial sampler outperforms all non-combinatorial sampler by 48.88\% (AUC measure) average overall datasets. }
 \label{fig:recall_independent_vs_combination}
\end{figure}


%

\Cref{fig:adult_heatmap} shows that when arms of two attributes are combined to generate a new query, it performs better than the two corresponding arms that are queried disjointly. The leftmost and the topmost sliders shows the quality of arms of individual attributes. The left matrix is the combinatorial queries of the two attributes: ``education'' and ``marital status''. The right matrix is the combination of two arms which are queried disjointly and their quality is measured via average quality of the two corresponding arms. Observe that the real-world arm combinations help us explore very high quality combinatorial arms (darker high quality regions defined by setting ``education'' to 9 or 10 and marital status set to 5). Notice that such types of correlation between attributes helps recover high yielding arms by the decision tree. However when the queryable attributes are independent, decision tree works identical to a naive MAB sampler.


\begin{figure}
 \centering
 \includegraphics[width=\linewidth,trim={1cm 0.3cm 0cm 0cm},clip]{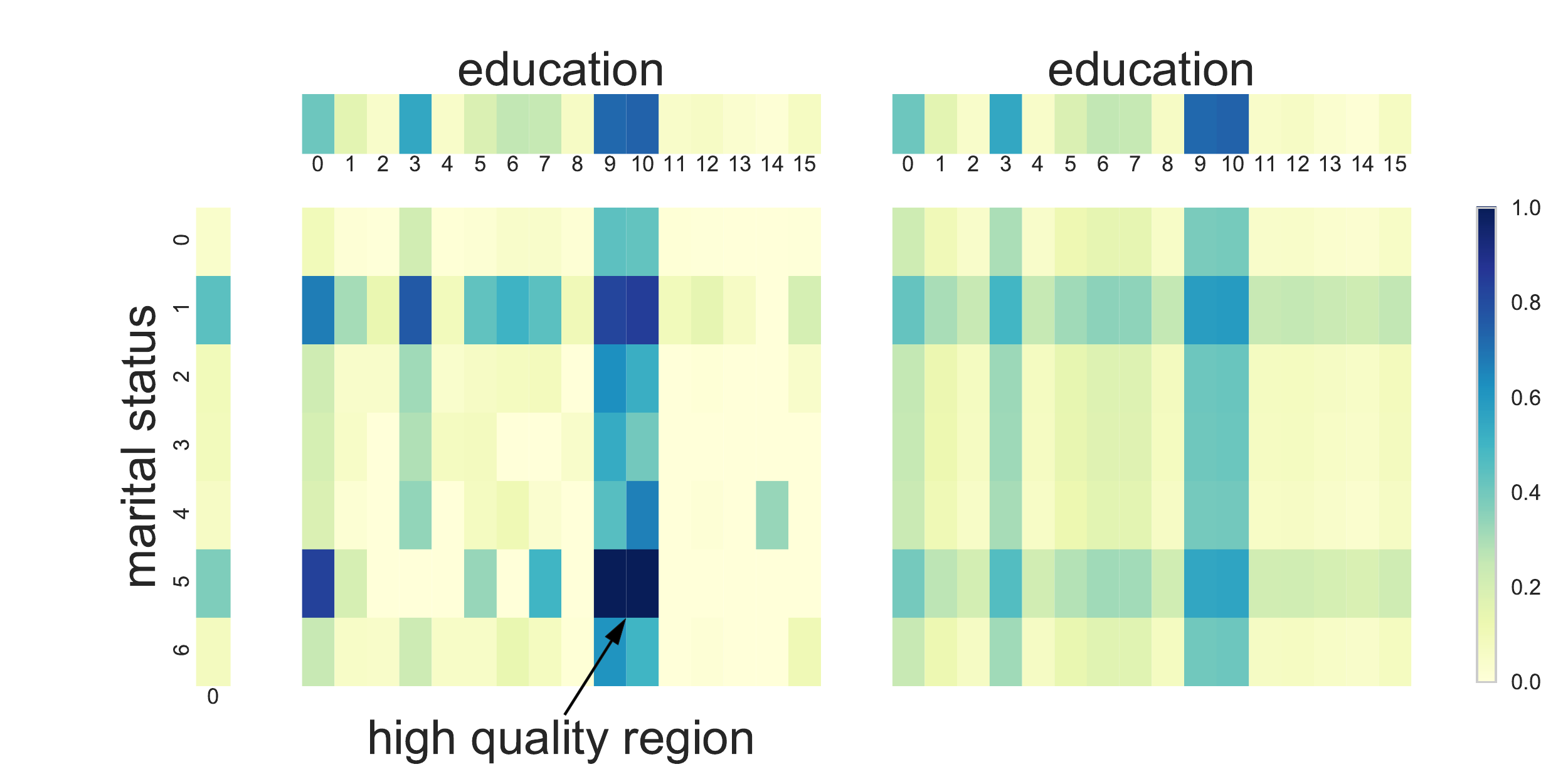}
 \caption{ Comparison between combinatorial queries and disjoint combination of queries shown via heat-map of two attributes, ``marital status'' and ``education'', in Adult dataset. Darker shade represents higher quality arms. }
 \label{fig:adult_heatmap}
\end{figure}

\subsection{Digging deeper: Factors affecting sampling}
\label{subsec:dig_deeper}
We now explore three prominent factors that impact hidden population sampling. This analysis will help users to be mindful of different factors that may affect their hidden population sampling.


\subsubsection{Effect of page size}

Page size is directly proportional to the performance of a sampler. Higher page size means larger number of samples obtained in every API call. More samples leads to higher recall values. ~\Cref{table:result_size_improvement} depicts the percentage improvement in recall value as a consequence of increasing the page size across various samplers. We observe that recall of \texttt{DT-TMP} is better than baseline samplers by a significant factor. Furthermore, we note \texttt{DT-TMP} has the second best improvement in recall value. \texttt{UNI} has the best improvement in recall value as the page size increases. This is due to the fact that \texttt{UNI} avoids the over-sampling issue as discussed later. However, it should be noted that \texttt{UNI} is typically not possible in practice since it is not supported by APIs of most real-world data sources such as Facebook, Twitter and LinkedIn.

For MAB samplers, we observe that increasing page size from 5 to 10 leads to \textit{diminishing returns}, which is used to refer to the phenomenon that the recall value increases by lesser than $\alpha \%$ as the page size increase by $\alpha \%$. The diminishing returns takes place due to two reasons. One, at high page size, generic queries are more likely to over-sample the same entities that are sampled by their specific queries. Two, we observe non-uniform skewed distribution of population over queryable attributes. Thus, non-uniform query sizes lead to under-sampling of entities. Under-sampling of entities refers to the scenario when the query size is less than page size, thus forcing the API to return less than page size results.

\begin{table}[h]
\center
\caption{Percentage improvement in recall value when page size increase from 5 to 10. UNI that samples uniformly over the database shows an expected improvement of nearly $100\%$. \texttt{DT-TMP} is the best performing sampler. However, we observe the effect of  diminishing results. At very high page size, all samplers behave similarly.}
\label{table:result_size_improvement}
\resizebox{\linewidth}{!}{%
\begin{tabular}{@{}c cc cc cc@{}}
 & \multicolumn{2}{c}{\textit{Patent}} & \multicolumn{2}{c}{\textit{Auto}}& \multicolumn{2}{c}{\textit{Adult}} \\ \toprule
Samplers& $R_{5}$& $\Delta R_{5\rightarrow 10}\%$& $R_{5}$& $\Delta R_{5\rightarrow 10}\%$& $R_{5}$& $\Delta R_{5\rightarrow 10}\%$\\ \toprule
\texttt{UNI} 	& 	1.39E-03 	& 	96.94 	& 	4.02E-04 	& 	93.39 	& 	5.43E-02 	& 	93.03 \\
\texttt{TMP} 	& 	4.96E-04 	& 	89.24 	& 	5.02E-04 	& 	58.20 	& 	1.92E-02 	& 	37.08 \\
\texttt{EXP} 	& 	7.50E-04 	& 	90.30 	& 	3.37E-04 	& 	49.85 	& 	1.69E-02 	& 	39.68 \\ \midrule
\texttt{RW} 	& 	6.13E-04 	& 	87.78 	& 	3.00E-04 	& 	22.95 	& 	1.04E-02 	& 	44.56 \\
\texttt{LS} 	& 	5.20E-04 	& 	138.12 	& 	1.91E-04 	& 	72.36 	& 	1.43E-02 	& 	69.58 \\
\texttt{CB} 	& 	2.17E-03 	& 	107.59 	& 	3.60E-04 	& 	147.54 	& 	7.89E-02 	& 	49.89 \\ \midrule
\texttt{DT-TMP} 	& 	5.34E-03 	& 	83.13 	& 	1.32E-03 	& 	60.47 	& 	1.22E-01 	& 	65.85 \\ \bottomrule
\end{tabular}
}
\end{table}

\subsubsection{Effect of number of queryable attributes}

\begin{figure}
\centering
\includegraphics[width=\linewidth,trim={3cm 0.3cm 5.5cm 1cm},clip]{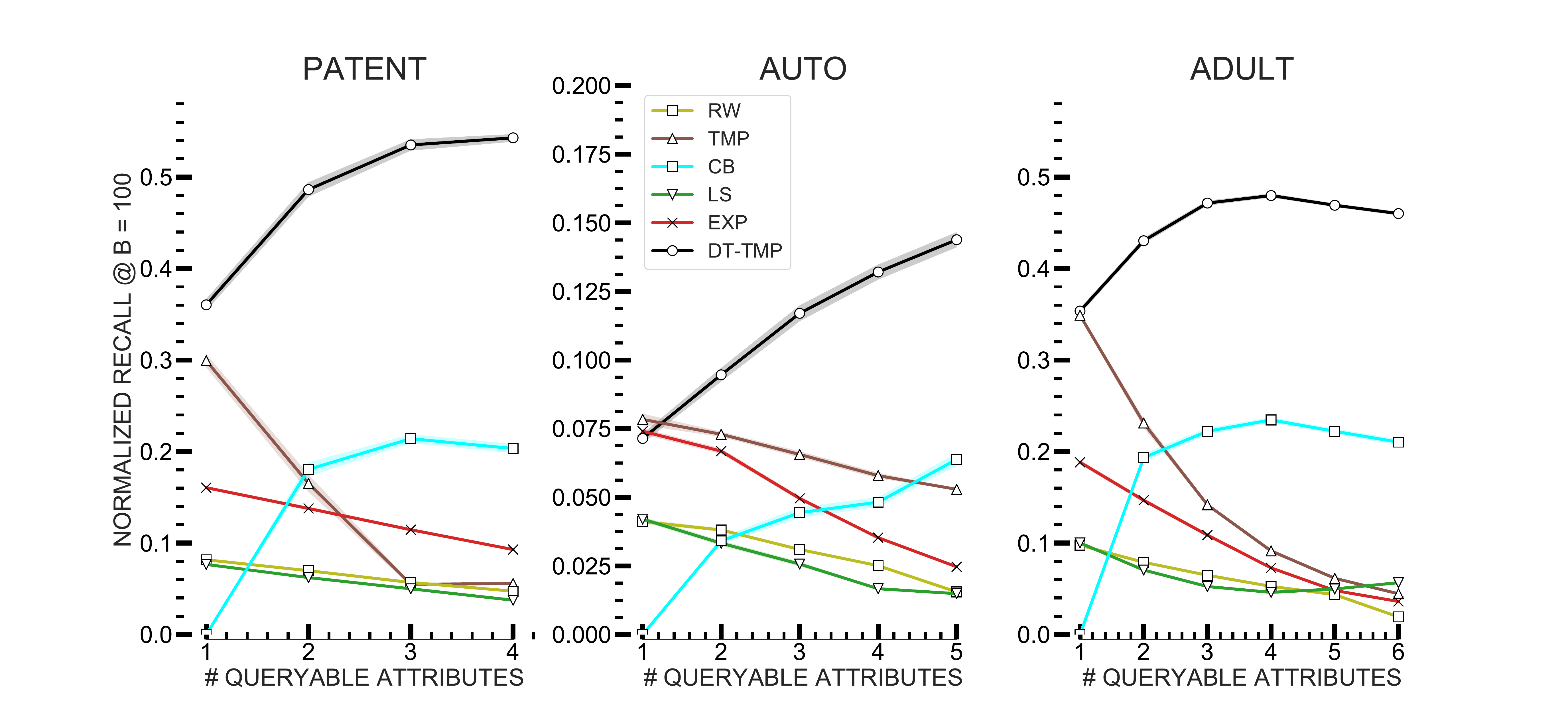}
\caption{Recall value of different samplers at budget of 100 API calls. As the number of queryable attributes increases, \texttt{DT-TMP}'s performance increases due to it's flexibility in exploring over large query space. Large number of attributes creates exponential more arms to explore for naive MAB samplers, thereby causing a drop in their performance. }
\label{fig:attrnums}
\end{figure}

We now explore the effect of number of queryable attributes on discover-ability of hidden target population. We try different subsets of queryable attributes described in \Cref{tab:dataset_stats} for the offline datasets to observe this effect. We average the results of attribute combinations (subsets) where equal number of queryable attributes are used. \Cref{fig:attrnums} shows the efficacy of \texttt{DT-TMP} over baseline samplers as the number of queryable attributes increases. Existing samplers suffer from high exploration space associated with large number of queryable attributes. \texttt{DT-TMP} applies the decision tree based search to preferentially explore high yielding queries in large query spaces. It therefore performs significantly better than baseline samplers. However when number of queryable attributes is one, it is observed in ``auto'' subplot of \Cref{fig:attrnums} that \texttt{DT-TMP} performs slightly worse than baseline samplers; for all other subplots the samplers are statistically indistinguishable at a confidence interval of 95\%. This happens because \texttt{DT-TMP} trades off the explore-exploit phase in typical MABs to greedily explore new attribute combinations. It therefore performs slightly poorly compared to baseline MAB samplers when the number of combinations is very few. Since UNI is independent of attributes, we do not include it in \Cref{fig:attrnums,fig:attrcardinality}.

%
%
%
%
%
%

\subsubsection{Effect of queryable attributes' cardinality}

Attribute cardinality controls the size of query space---higher attribute cardinality implies larger query space. We simulate the attribute cardinality in real-world datasets by modifying the attribute cardinalities in \Cref{tab:dataset_stats} to fixed number say $c$. For each attribute, we preserve the top yielding $c-1$ attribute-values and merge the low yielding remaining attribute-values into a new attribute-value. The merging preserves the discover-ability of hidden entities in the dataset over each attribute. We observe the effect of varying attribute cardinality in \Cref{fig:attrcardinality}. \texttt{DT-TMP} owing to it's search tree strategy is very effective at exploring and exploiting rewards distributed over large query space which is induced by high cardinality attributes.

\begin{figure}
\centering
\includegraphics[width=\linewidth,trim={3cm 0.3cm 5.5cm 1cm},clip]{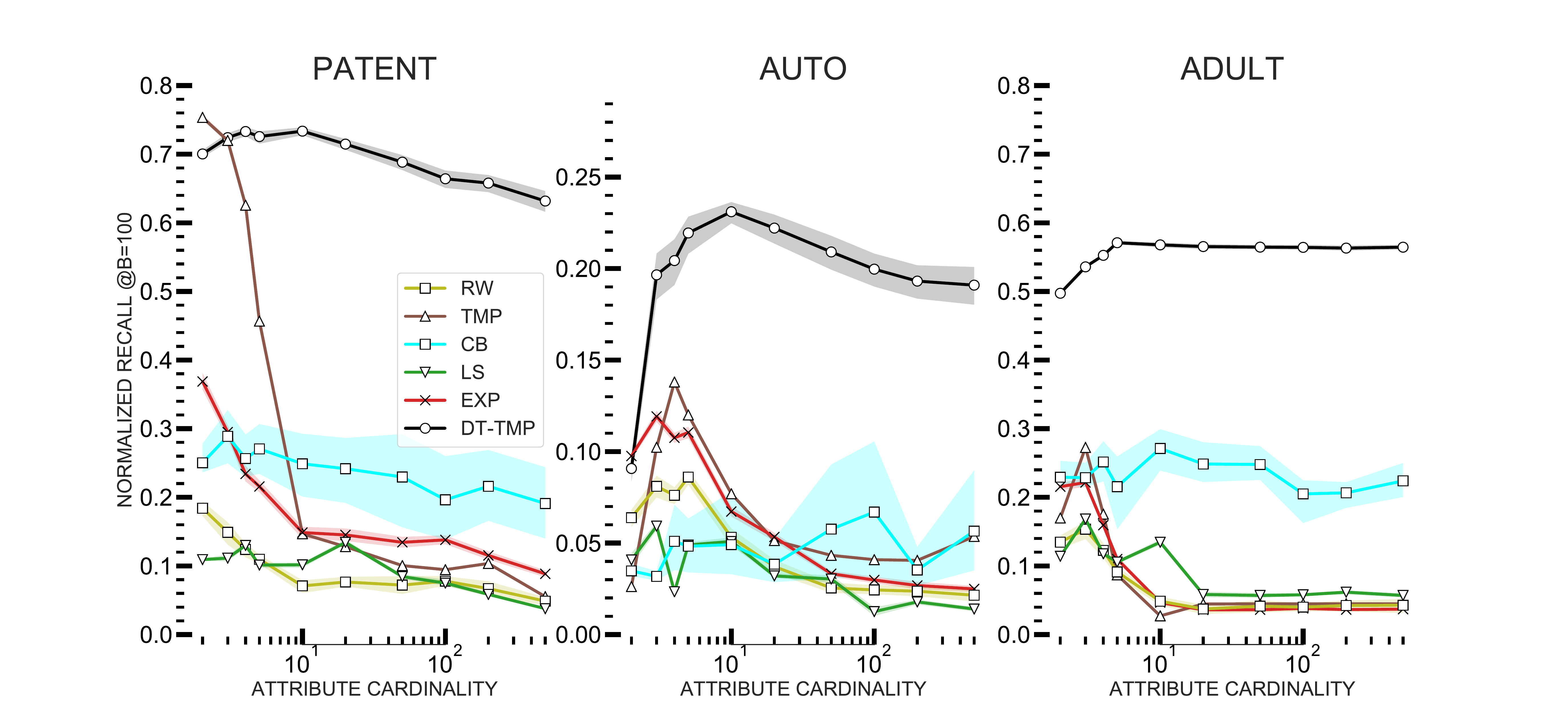}
\caption{Recall value for different sampling strategies at budget of 100 API calls across variable attribute cardinality. Combinatorial MAB (\texttt{DT-TMP}) does increasingly better over large attribute space created by high attribute cardinalities. }
\label{fig:attrcardinality}
\end{figure}

\subsubsection{Effect of correlation between queryable attributes and the hidden property}

Correlation between the queryable attributes and the hidden property is one of the most important metrics. Highly correlated queryable attributes help form queries that yield high number of hidden population entities and vice versa. We control the vary between queryable attributes and hidden attributes by randomly shuffling a fractional (shuffle ratio) subset of the dataset. \Cref{fig:correlation} shows that the performance of all samplers fall as the shuffle ratio increases. This happens because the at high shuffle rate the hidden population entities are uniformly distributed across different queries of the population, thus leading to poor recall values. The relatively higher performance of \texttt{DT-TMP} over the baseline samplers at even higher shuffle ratio is on account of the fact that \texttt{DT-TMP} avoids re-sampling of entities. 

\begin{figure}
\centering
\includegraphics[width=\linewidth,trim={3cm 0.3cm 5.5cm 1cm},clip]{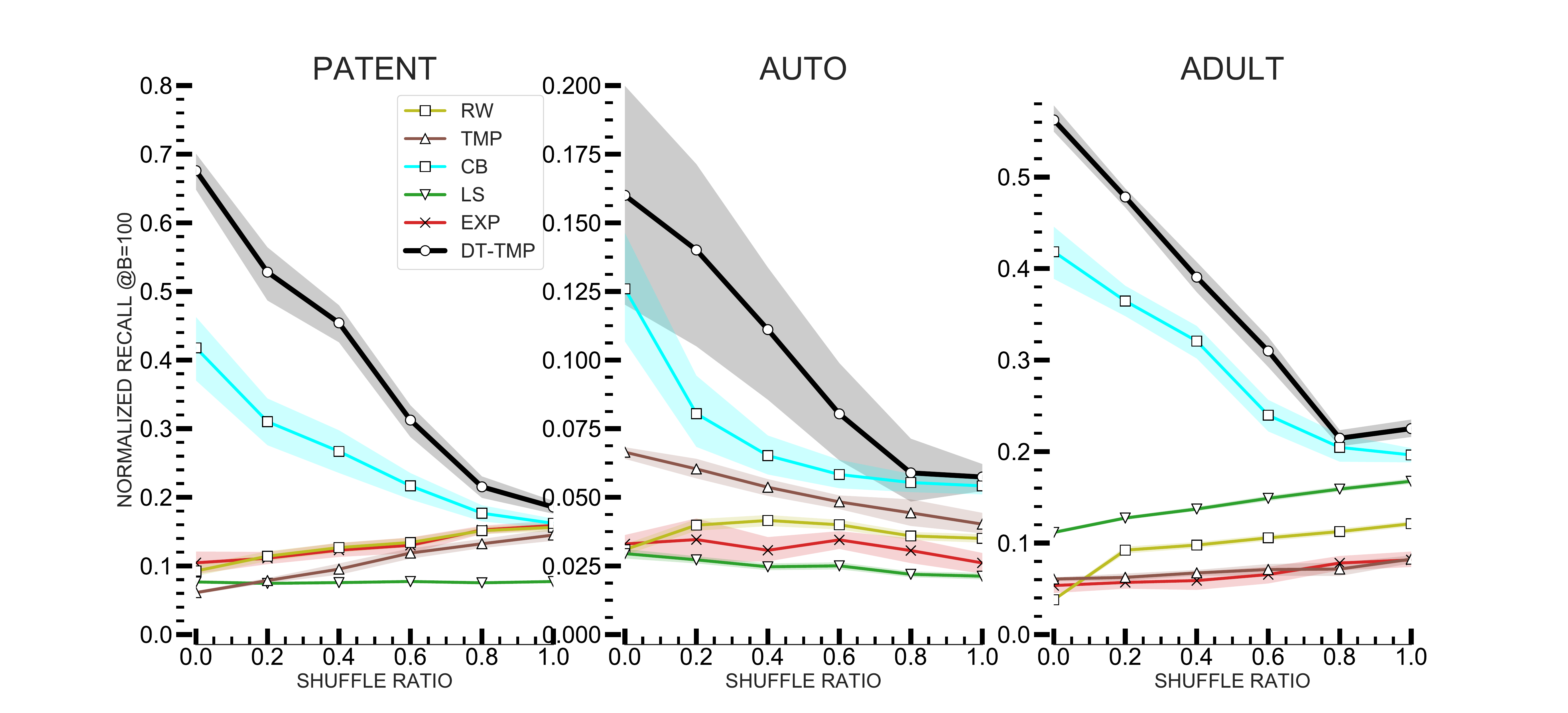}
\caption{Recall value for different sampling strategies at budget of 100 API calls across different shuffle rates. The sampling performance falls with increasing shuffle rate. The performance depicts that \texttt{DT-TMP} is better at learning the correlation between hidden property and queryable attributes.  }
\label{fig:correlation}
\end{figure}

\subsection{Limitations}
\label{subsec:limitation}
Now, we discuss four limitations of this work. First, our algorithms are agnostic to the ranking function and therefore ordering of the results. However, when the ranking function is correlated with the hidden target attribute, agnostic samplers may not perform well. We note that non-stationary reinforcement learners should be used to handle the effect of non-stationary reward induced by unknown ranking functions. Second, our algorithms rely upon a classifier for identifying target population when the hidden target attribute is implicit. Our future work is to discern the effect of classifier(s) in hidden population sampling. Third, our MABs handle continuous and infinite valued attributes by discretizing the attribute to ensure that there is finite arms of MAB. Third, much of our work assumes content to be static. In future, it will therefore be useful to modify the \texttt{DT-TMP} sampling strategy to handle streaming datasets.



\section{Existing work}
\label{sec:related work}

In this section, we discuss various methods that are closely related to the problem of sampling hidden population from OSNs.

Focused crawling is a well-studied problem wherein a crawler tries to maximize the coverage of a given target topic such as ``semiconductor related web-pages'' by traversing web-links. Chakraborti et al.~\cite{chakrabarti1999focused} proposed focused crawler that iteratively explores web-links that are more likely to fetch topic web-pages. Similar works on crawling are focused around exploiting the information of web-page such as it's content link structure, URL and metadata~\cite{menczer2000adaptive,hersovici1998shark} to efficiently crawl target pages. In contrast to the focused crawling that uses graph based interface, our samplers use a form based interface to iteratively query for a hidden population.  


Hidden web crawling is an area of research that tries to gather the entire population or database contents by efficiently querying or crawling via database's interface. Raghavan et al \cite{raghavan2000crawling} first proposed a task specific hidden web crawler called as Hidden Web Exposer that crawled the hidden web forms by maintaining Label Value Set table used for filling out the forms. Wu et al\cite{wu2006query} proposed attribute value graph traversal based heuristics to crawl the hidden databases. Several other works such as \cite{zheng2013learning} tries to sample hidden database entirely in fewest number of queries. Sheng et al\cite{sheng2012optimal} showed optimal algorithms for crawling entire hidden web database from form based interfaces. However, in contrast to the previous studies in hidden web crawling that aims at discovering the \textit{entire} hidden database, our sampler is focused towards a \textit{target} hidden population. Furthermore, most of the previous works have either been been limited to textual query interfaces or require a prior knowledge or seed set of the well defined topic~\cite{liakos2016focused}. Another limitation of existing form based interfaces is that they consider all results pertaining to a query to be obtained in a single API query \cite{nazi2015querying,sheng2012optimal} which is not a realistic assumption for most web interfaces such as GitHub and Twitter.

Query reformulation is another line of research that works at identifying better queries for higher recall in text retrieval systems. Existing retrieval systems in literature are predominantly designed to search for new queries that yield higher reward \cite{li2016multiple}. Query reformulation systems \cite{rieh2006analysis} typically rewrites a query to find a new query that maximizes the number of relevant document returned. In contrast to the query retrieval systems that retrieves documents usually expressed by rich textual information, our work focuses is entity retrieval problem where discovered entities provide very limited information in form of few attributes and the query interface is limited to queryable attributes of the entities. To the best of the knowledge, this is the first work that aims at retrieving hidden target entities in OSNs by querying their faceted APIs using attribute combinations.


\section{Conclusion and Future Work}
\label{sec:conclusion}

This paper proposed a novel algorithm for sampling hidden target populations from online social networks. However, sampling individuals from hidden populations is hard due to API rate limits, limited access methods (or limited number of queryable attributes) and combinatorial query space to search from. To address these challenges, we modeled the problem as a Multi-Armed Bandit problem. We proposed a state-aware \texttt{DT-TMP} that exploited structure in combinatorial query space to discover high yielding queries. Our proposed sampler is better than the competing samplers by factor of 0.9-1.5$\times$ on offline real-world datasets where query size is returned by the API. Our samplers perform by margin of 54\% on Twitter hidden population tasks and 49\% on RateMD experiments. Exploring the effect of classifiers in discovering hidden populations is the focus of our future work.

\printbibliography

\end{document}